\documentclass[submission,copyright,creativecommons]{eptcs}
\providecommand{\event}{AiML 2026} 

\usepackage{aiml26}

\usepackage{iftex}

\usepackage{amssymb, stmaryrd, MnSymbol}
\usepackage{enumerate}
\newtheorem{convention}[theorem]{Convention}

\ifpdf
  \usepackage{underscore}         % Only needed if you use pdflatex.
  \usepackage[T1]{fontenc}        % Recommended with pdflatex
\else
  \usepackage{breakurl}           % Not needed if you use pdflatex only.
\fi

\title{Goldblatt-Thomason Theorem for Probability Logic}
\author{Somayeh Chopoghloo$^{1}$ \qquad Massoud Pourmahdian$^{1,2}$ \qquad Reihane Zoghifard$^{1}$
	\institute{$^{1}$School of Mathematics, Institute for Research in Fundamental Sciences (IPM), Tehran, Iran}\thanks{The first author gratefully acknowledges the support of the Iran National Science Foundation (INSF) under Grant No. 4013209 for supporting this project.}\\
	\institute{$^{2}$Department of Mathematics and Computer Science, \\Amirkabir University of Technology (Tehran Polytechnic), Tehran, Iran}
	\email{s.chopoghloo@ipm.ir \quad\quad\quad\qquad pourmahd@ipm.ir \quad\quad\qquad\quad r.zoghi@gmail.com}
}

\newcommand{\titlerunning}{Goldblatt-Thomason Theorem for Probability Logic}
\newcommand{\authorrunning}{S. Chopoghloo, M. Pourmahdian \& R. Zoghifard}

\hypersetup{
  bookmarksnumbered,
  pdftitle    = {\titlerunning},
  pdfauthor   = {\authorrunning},
  pdfsubject  = {EPTCS},               % Consider adding a more appropriate subject or description
  pdfkeywords = {Goldblatt-Thomason Theorem, Markov processes, (modal) probability logic, model theory} % Uncomment and enter keywords specific to your paper
}

\usepackage{xcolor,color}

\begin{document}
\maketitle

\begin{abstract}
Probability logic ($\mathsf{PL}$) extends propositional logic with countably many probability operators, one for each rational number between 0 and 1. The formulas of this logic are interpreted over the class of Markov processes, i.e., structures of the form $\langle \Omega, \Sigma, T \rangle$, where $\langle \Omega, \Sigma \rangle$ is a measurable space and $T$ is a Markov kernel.  

The main contribution of this paper is the establishment of the Goldblatt-Thomason theorem for probability logic. As an application, we show that the class of Harsanyi type spaces is definable in $\mathsf{PL}$. Moreover, we obtain some variants of the Goldblatt-Thomason theorem for specific subclasses of Markov processes.
\end{abstract}
%%%%%%%%%%%%%%%%%%%%%%%%%%%%%%%%%%%%%%%%
%%%%%%%%%%%%%%%%%%%%%%%%%%%%%%%%%%%%%%%%
\section{Introduction}

Probability logic aims to provide a formal framework for studying probabilistic systems, such as Markov processes. It supports reasoning about quantitative statements of the form `an event occurs with probability at least $r$' and has been widely applied in diverse areas, including theoretical computer science %it is used in the study of stochastic and labelled Markov processes 
\cite{deshar:bisim02,panan:lmp09}, artificial intelligence %it underlies logical approaches to reasoning about knowledge and uncertainty 
\cite{fagin:logic90}, and economics %it appears in the analysis of type spaces in game theory 
\cite{Aum1999,heifmon:prob01}.

There are different syntactic frameworks for studying probabilistic systems, such as first-order syntax (see, e.g., \cite{keisler}) which may provide a more expressive language, or modal syntax (see, e.g., \cite{fagin:logic90,heifmon:prob01}) which is less expressive but decidable. Throughout this paper, we investigate the modal framework. Probability logic is a multimodal extension of propositional logic obtained by adding a family of modal operators $\{L_r \mid r \in \mathbb{Q} \cap [0,1]\}$. The intended interpretation of a formula $L_r \varphi$ is that `the probability of $\varphi$ is at least $r$'.

Markov processes provide a mathematically rigorous model for probabilistic state-based systems, particularly in theoretical computer science. A (labeled) Markov process consists of a measurable state space equipped with transition probability kernels (indexed by actions), thereby modeling systems that combine probabilistic evolution with observable interaction. Such structures serve, for example, as denotational models for probabilistic programming languages and as semantic foundations for verification, equivalence, abstraction, and compositional reasoning in stochastic systems. Formally, a Markov process is a structure $\langle\Omega,\Sigma,T\rangle$, where $\langle\Omega,\Sigma\rangle$ is a measurable space and $T: \Omega \times \Sigma \to [0,1]$ is a Markov kernel.

In this paper, we investigate definability of classes of Markov processes by means of theories in probability logic. Our main result establishes a version of the Goldblatt-Thomason theorem for this setting. The classical theorem of Goldblatt and Thomason characterizes the modally definable elementary classes of Kripke frames in terms of certain frame constructions. It states that an elementary (first-order definable) class of Kripke frames is modally definable if and only if it is closed under disjoint unions, generated subframes, bounded morphic images, and reflects ultrafilter extensions. The original proof is algebraic in nature, relying on representation theory for modal algebras. It is further shown that a class of general frames-- i.e., Kripke frames equipped with a Boolean algebra of admissible sets--is modally definable precisely when it is closed under disjoint unions, generated subframes, bounded morphic images, and both closed under and reflective of ultrafilter extensions. Subsequently, van Benthem provided a model-theoretic proof of the original Goldblatt-Thomason theorem \cite{benthem:revisited93}.

Goldblatt-Thomason style characterizations have since been established for numerous extensions of modal and intuitionistic logics, including graded modal logic \cite{sano:gtgraded10}, intuitionistic logic \cite{rodenburg2016intuitionistic,sano2020goldblatt,de2020goldblatt,ma2025goldblatt}, and coalgebraic modal logic \cite{kurz:golcoal07}. In particular, van Benthem studied definable classes of finite (transitive) frames \cite{van1988notes}. He showed that for every finite transitive point-generated frame $\mathfrak{F}$ there exists a formula $\varphi_{\mathfrak{F}}$, known as a Jankov-Fine formula, such that whenever $\varphi_{\mathfrak{F}}$ is satisfiable on a frame $\mathfrak{G}$, there exists a surjective bounded morphism from a generated subframe of $\mathfrak{G}$ onto $\mathfrak{F}$. Variants of such Jankov-Fine constructions have been employed to characterize definable classes of finite structures in other logical settings, for example, in conditional logic over finite posets \cite{fornasiere2024frame}.

The study of definability in probability logic presents substantial additional difficulties. A central obstacle is the failure of compactness. The classical Goldblatt-Thomason theorem is formulated for elementary classes and can be weakened in terms of closure under ultrapowers. In contrast, satisfaction of probability formulas is not preserved under ultraproducts, and compactness fails accordingly. Our approach therefore follows the version for general frames and replaces closure under ultrapowers with ultrafilter extensions. However, defining suitable ultrafilter extensions for Markov processes is technically non-trivial. To address this issue, we adapt ideas from \cite{kozen2013stone}, working with Markov processes whose $\sigma$-algebras admit countable generating sets.
%Indeed, in \cite{kozen2013stone}, the authors studied Stone duality for Markov processes. To investigate the algebraic counterpart of the Markov processes, which is called Aumann algebras, they introduced an ultrafilter construction over a specific kind of Markov processes, named Stone–Markov process.
Indeed, Kozen and coauthors \cite{kozen2013stone} studied Stone duality for Markov processes. In order to analyze the corresponding algebraic structures, known as Aumann algebras, they introduced an ultrafilter construction on a particular class of Markov processes, termed Stone-Markov processes. 
They proved the duality between countable Aumann algebras and Stone-Markov Markov processes.
Accordingly, we also consider the class of Markov processes with countably generated $\sigma$-algebras, and employ methods similar to those in \cite{kozen2013stone} to define ultrafilter extensions of Markov processes.

In addition to the general case, we investigate definability over the class of finite Markov processes, obtaining a corresponding characterization for this important special case.

The outline of this article is as follows. In Section~\ref{sec per}, a brief review of the basic notions and facts needed for other parts is given.
In Section~\ref{sec model}, model-theoretic constructions used in the Goldblatt-Thomason theorem are introduced. 
In Section~\ref{sec gt}, a version of the Goldblatt-Thomason theorem is established.
Finally, in Section~\ref{sec variants}, a version of the Goldblatt-Thomason theorem for the class of finite Markov processes is given.
%Furthermore, some proofs are given in the appendix.
%%%%%%%%%%%%%%%%%%%%%%%%%%%%%%%%%%%%%%%%
%%%%%%%%%%%%%%%%%%%%%%%%%%%%%%%%%%%%%%%%
\section{Preliminaries}\label{sec per}

In this section, we review the basic materials needed for the main results of this paper.

Let $\mathbb{P}$ be a countable set of propositional variables. The formulas of probability logic ($\mathsf{PL}$) are defined recursively by the following grammar:
\begin{align*}
	\varphi ::&=\ \ p\mid \neg\varphi \mid \varphi \land \varphi \mid L_r \varphi
\end{align*}
where $p \in \mathbb{P}$ and $r\in \mathbb{Q}_0\ := [0,1] \cap \mathbb{Q}$. The Boolean connectives $\top$, $\perp$, $\to$, $\leftrightarrow$, and $\vee$ are defined in the usual way. For a formula $\varphi$, the expression $L_r \varphi$ is interpreted as `the probability of $\varphi$ is at least $r$'. Likewise, $M_r \varphi$ is defined as an abbreviation for $L_{1-r}\neg \varphi$ and is interpreted as `the probability of $\varphi$ is at most $r$'.   Moreover, for $r_1,\dots,r_k \in \mathbb{Q}_0$, we use the formula $L_{r_1\dots r_k} \varphi$ as an abbreviation for $L_{r_1} \dots L_{r_k} \varphi$.
Additionally, $L_r^0 \varphi := \varphi$ and $L_r^n \varphi:= \underbrace{ L_r \; \dots\; L_r}_\text{$n$-times} \varphi$ for each $n>0$ and $r\in \mathbb{Q}_0$.

\begin{definition} \label{Markov}
	Let $\langle{\Omega, \Sigma}\rangle$ be a measurable space. A function $T: \Omega\times \Sigma\to [0, 1]$ is called a {\em Markov kernel} (or {\em transition probability}) if the following conditions hold:
	\begin{itemize}
		\item for each $w \in \Omega$, the mapping $T(w, .): N \mapsto T(w, N)$ is a probability measure on $\Sigma$;
		\item for each $N \in \Sigma$, the mapping $T(., N): w \mapsto T(w, N)$ is a measurable function on $\Omega$.
	\end{itemize}
	The triple $\langle{\Omega, \Sigma, T}\rangle$ is then called a {\em Markov process} on the state space $\langle{\Omega, \Sigma}\rangle$. For convenience, we may write $T(w)(N)$ instead of $T(w, N)$.
\end{definition}

\begin{definition}\label{coun gen def}
	A Markov process $\mathfrak{P}= \langle{\Omega, \Sigma, T}\rangle$ is said {\em (countably) generated} by a (countable) subalgebra $\mathcal{B}$ of $\Sigma$ if  $\mathcal{B}$ is a (countable) Boolean algebra and $\Sigma:= \sigma(\mathcal{B})$,  the smallest $\sigma$-algebra containing $\mathcal{B}$. We always assume that $F_r(N) := \{w\in\Omega\;|\; T(w)(N)\geq r\}\in \mathcal{B}$, for each $N \in \mathcal{B}$ and $r\in \mathbb{Q}_0$. We also define $F_{r_1\dots r_n}(N)$ as $F_{r_1}\dots F_{r_n}(N)$,  for $r_1, \dots, r_n\in \mathbb{Q}_0$. 
\end{definition}

\begin{convention} \label{con}
	If we represent the Markov process $\mathfrak{P}$ as $ \langle{\Omega, \mathcal{B}, T}\rangle$, then we mean that $\Sigma= \sigma(\mathcal{B})$.
\end{convention}

\begin{definition}
	A {\em Markov model}, or simply a {\em model}, over (or based on) a Markov process $\mathfrak{P}= \langle{\Omega, \mathcal{B}, T}\rangle$ is a tuple $ \mathfrak{M} = \langle{\Omega, \mathcal{B}, T, v}\rangle$ where $v: \mathbb{P} \to \mathcal{B}$ is a valuation function which assigns to every propositional variable $p \in \mathbb{P}$, a set $v(p)\in\mathcal{B}$.
\end{definition}

\begin{remark}
Note that restricting the valuation to the generating subalgebra $\mathcal{B}$ is analogous to defining valuations on general Kripke frames, in which each propositional variable $p$ is assigned a value in the augmented algebra associated with the frame. This restriction is necessary in order to adapt the method used in proving the Goldblatt--Thomason theorem for general frames to the setting of Markov processes.
\end{remark}

\begin{definition}
	The {\em satisfaction relation} for arbitrary formulas  in a given model $ \mathfrak{M} = \langle{\Omega, \mathcal{B}, T, v}\rangle$, is
	defined inductively as follows:
	\begin{itemize}
		\item[] $\mathfrak{M}, w \vDash p \;$ iff $\; w \in v(p)$,
		\item[] $\mathfrak{M}, w \vDash \neg\varphi \;$ iff $\; \mathfrak{M}, w \nvDash \varphi$,
		\item[] $\mathfrak{M}, w \vDash \varphi\land \psi\;$ iff $\; \mathfrak{M}, w \vDash \varphi$ and $\mathfrak{M}, w \vDash \psi$,
		\item[] $\mathfrak{M}, w \vDash L_r \varphi \;$ iff $\; T(w, [\![\varphi]\!]_{\mathfrak{M}}) \geq r$, where
		$[\![\varphi]\!]_{\mathfrak{M}} = \{ w \in \Omega\;| \; \mathfrak{M}, w \vDash \varphi \}$.
	\end{itemize}
	A set $\Gamma$ of formulas is satisfiable in the world $w$ if $\mathfrak{M}, w \vDash \varphi$ for each $\varphi \in \Gamma$.
\end{definition}

Note that by Convention \ref{con}, we have $[\![\varphi]\!]_{\mathfrak{M}}\in\mathcal{B}$, for each formula $\varphi$.

\begin{definition}
	A formula $\varphi$ is {\em valid in a model} $\mathfrak{M}= \langle{\Omega, \mathcal{B}, T, v}\rangle$, denoted by $\mathfrak{M}\vDash \varphi$, if $\mathfrak{M}, w \vDash \varphi$ for all $w \in \Omega$. Similarly, $\varphi$ is \textit{valid in a Markov process} $\mathfrak{P} = \langle{\Omega, \mathcal{B}, T}\rangle$, denoted by $\mathfrak{P}\vDash \varphi$, if it is valid in every model based on $\mathfrak{P}$. We say that $\varphi$ is \textit{valid in a class $\mathcal{C}$} of Markov processes, denoted by ${\vDash}_{\mathcal{C}}\; \varphi$, if it is valid in every element $\mathfrak{P}=\langle{\Omega, \Sigma, T}\rangle$ of $\mathcal{C}$. %Further, $\varphi$ is a {\em semantic consequence} of a set $\Gamma$ of formulas over a class $\mathcal{C}$ if for every model $\mathfrak{M}$ based on an element in $\mathcal{C}$, if $\Gamma$ holds in a world $w$ in $\mathfrak{M}$, then so is $\varphi$. In this case, we write $\Gamma \;{\vDash}_{\mathcal{C}} \;\varphi$.
	For abbreviation, we may omit the subscript $\mathcal{C}$ and instead write ${\vDash}\; \varphi$ when it is the class of all Markov processes.
	A set of formulas $\Gamma$ is satisfiable in the class $\mathcal{C}$ if there is a Markov process $\mathfrak{P}\in\mathcal{C}$ and a model $\mathfrak{M}$ based on $\mathfrak{P}$ such that $\mathfrak{M},w\models\Gamma$ for some $w\in\Omega$.
	Likewise, $\Gamma$ is finitely satisfiable in $\mathcal{C}$ if each finite subset $\Gamma'\subseteq\Gamma$ is satisfiable in $\mathcal{C}$.
\end{definition}

\begin{proposition} \label{validity}
	Let  $\varphi$ be a formula and $\mathfrak{P} = \langle{\Omega, \Sigma, T}\rangle$  a Markov process. Then the following conditions are equivalent:
	\begin{enumerate}
		\item $\varphi$ is valid in $\mathfrak{P} = \langle{\Omega, \Sigma, T}\rangle$.
		\item $\varphi$  is valid in $\mathfrak{P} = \langle{\Omega, \mathcal{B}, T}\rangle$, for every subalgebra $\mathcal{B}$ of $\Sigma$ with $\Sigma= \sigma(\mathcal{B})$.
		\item $\varphi$ is valid in every $\mathfrak{P}' = \langle{\Omega, \Sigma', T'}\rangle$, where $\Sigma'$ is a countably generated $\sigma$-subalgebra of $\Sigma$ and $T' = T \upharpoonright\Sigma' $.
		\item $\varphi$ is valid in every  $\mathfrak{P}' = \langle{\Omega, \mathcal{B}', T'}\rangle$, where $\mathcal{B}'$ is a countable subalgebra of $\Sigma$ and $T'= T\upharpoonright \sigma(\mathcal{B}')$.
	\end{enumerate}
\end{proposition}
\begin{proof}%[Proof of Proposition \ref{validity}]
	The equivalence of Parts 1 and 2, and respectively, Parts 3 and 4 are obvious. To see that Part 4 implies Part 1, notice that if $\varphi$  is not valid in $\mathfrak{P} = \langle{\Omega, \Sigma, T}\rangle$, then there exists a valuation $v:\mathbb{P}\to \Sigma$ such that in $\mathfrak{M} = \langle{\mathfrak{P}, v}\rangle$, we have that  $\mathfrak{M}, w \nvDash \varphi$, for some  $w\in \Omega$. Now if we take $\mathcal{B}' = \{[\![\theta]\!]_{\mathfrak{M}}\;|\; \theta\; \text{is a formula}\}$, then $\varphi$ is not valid in $\mathfrak{P}'= \langle{\Omega, \mathcal{B}', T\upharpoonright \sigma(\mathcal{B}')}\rangle$.
	Finally, It is clear that Part 2 implies Part 4.	
\end{proof}

\begin{definition}
	A class of Markov processes $\mathcal{C}$ is definable in probability logic or \emph{P-definable}, if there is a set of formulas $\Gamma$ such that for any Markov process $\mathfrak{P}$ we have $\mathfrak{P}\in\mathcal{C}$ if and only if $\mathfrak{P}\models\Gamma$.
\end{definition}

\begin{definition} \cite[Definition~10]{kozen2013stone}
	Let $ \langle{\Omega, \mathcal{B}, T}\rangle$ be a countably generated Markov process. We say that this process is a {\em Stone-Markov process} if $\mathcal{B}$ is a countable base of clopen sets for the topology generated by $\mathcal{B}$, named $\tau_{\mathcal{B}}$, that satisfies the Baire property (also called \em topologically complete).
\end{definition}

Note that if  $\langle\Omega, \tau_\mathcal{B}\rangle$ is a Polish space or a compact separable space, then $ \langle{\Omega, \mathcal{B}, T}\rangle$ is a Stone-Markov process. This property is referred to as {\em saturation} in \cite{kozen2013stone}.
%%
%%
%%%%%%%%%%%%%%%%%%%%%%%%%%%%%%%%%%%%%%%%
%%%%%%%%%%%%%%%%%%%%%%%%%%%%%%%%%%%%%%%%
\section{Some Model-Theoretic Concepts}\label{sec model}
In this section, we introduce several model-theoretic constructions that are needed for the proof of the Goldblatt–Thomason Theorem for probability logic.

We begin by defining the notion of a Markov sub-process.

\begin{definition}[Markov sub-process]\label{subpro}
	Let $\mathfrak{P} = \langle \Omega, \Sigma, T\rangle$ be a Markov process, and let $\Omega' \in \Sigma$ be a measurable subset such that
	\[
	T(w,\Omega') = 1 \quad \text{for all } w \in \Omega'.
	\]
	Let $\Sigma' = \{ N \cap \Omega' \mid N \in \Sigma \}$ be the induced $\sigma$-algebra on
	$\Omega'$. Define the transition function $T' \colon \Omega' \times \Sigma' \to [0,1]$ by
	\[
	T'(w, N \cap \Omega') = T(w,N),
	\]
	for all $w \in \Omega'$ and $N \in \Sigma$.
\end{definition}

\begin{lemma}\label{lem sub}
	Under the assumptions of Definition~\ref{subpro}:
	\begin{enumerate}
		\item the triple $\mathfrak{P}' = \langle \Omega', \Sigma', T' \rangle$ is a Markov process;
		\item if $\mathfrak{P}$ is countably generated, then so is $\mathfrak{P}'$.
	\end{enumerate}
\end{lemma}
\begin{proof}%[Proof of Lemma~\ref{lem sub}]
	\begin{enumerate}
		\item The main point is to show that $T'$ is well defined. Let $N_1, N_2\in\Sigma$ and $N_1=N_2$. Since $T(w,\Omega') = 1$ for all $w \in \Omega'$, we have
		\[
		T'(w,N_1 \cap \Omega') = T(w,N_1) = T(w,N_2) = T'(w,N_2 \cap \Omega'),
		\]
		so $T'$ is well defined. Standard arguments show that $T'$ is a Markov kernel.
		
		\item Suppose that $\mathcal{B}$ is a countable Boolean algebra generating $\Sigma$.
		Then
		\[
		\mathcal{B}' = \{ B \cap \Omega' \mid B \in \mathcal{B} \}
		\]
		generates $\Sigma'$. Moreover, for each $r \in [0,1]$ and $B \in \mathcal{B}$,
		\[
		F_r(B \cap \Omega') = \{ w \in \Omega' \mid T'(w,B \cap \Omega') \ge r \}
		= F_r(B) \cap \Omega'.
		\]
		Thus $\mathfrak{P}'$ is countably generated.
	\end{enumerate}
\end{proof}

Once the notion of a sub-process has been established, we can define the corresponding notion of a sub-model.

\begin{definition}
	Let $\mathfrak{M} = \langle \Omega, \Sigma, T, v \rangle$ be a Markov model. Suppose that $\Omega' \in \Sigma$ satisfies the conditions of Definition~\ref{subpro}, and let $\mathfrak{P}' = \langle \Omega', \Sigma', T' \rangle$ be the corresponding sub-process. Define a valuation $v'$ by setting
	\[
	v'(p) = v(p) \cap \Omega' \quad \text{for each } p \in \mathbb{P}.
	\]
	Then $\mathfrak{M}' = \langle \Omega', \Sigma', T', v' \rangle$ is called a
	\emph{sub-model} of $\mathfrak{M}$.
\end{definition}

\begin{lemma}
	Let $\mathfrak{M}'$ be a sub-model of $\mathfrak{M}$. Then for every formula $\varphi$
	and every $w \in \Omega'$,
	\[
	\mathfrak{M}, w \vDash \varphi
	\quad \text{iff} \quad
	\mathfrak{M}', w \vDash \varphi.
	\]
\end{lemma}

%\begin{proof}
%	By induction on the structure of $\varphi$.
%\end{proof}

\begin{corollary}\label{sub val}
	Let $\varphi$ be a formula. If $\mathfrak{P} \vDash \varphi$ and
	$\mathfrak{P}'$ is a sub-process of $\mathfrak{P}$, then
	$\mathfrak{P}' \vDash \varphi$.
\end{corollary}
\begin{proof}%[Proof of Corollary \ref{sub val}]
	Suppose not. Then there exists a valuation $v'$ on $\mathfrak{P}'$ and a state
	$w \in \Omega'$ such that $\mathfrak{M}', w \nvDash \varphi$.
	Extend $v'$ to a valuation $v$ on $\mathfrak{P}$ by setting $v(p)=v'(p)$ for all
	$p \in \mathbb{P}$. By the previous lemma,
	$\mathfrak{M}, w \nvDash \varphi$, contradicting $\mathfrak{P} \vDash \varphi$.
\end{proof}

For a given model $\mathfrak{M}$ the following lemma introduces a sub-model $\mathfrak{M}'$ of $\mathfrak{M}$ which plays the same role as a {\em generated sub-model} in the context of the Kripke semantics.

\begin{lemma} \label{subL1}
	Suppose $\mathfrak{M} = \langle \Omega, \Sigma,T, v\rangle$ is a model, $w\in \Omega$ and $\Gamma$ is a set of formulas such that  $\mathfrak{M}, w \vDash L_1^n \varphi$  for each $\varphi\in \Gamma$ and $n\geq 0$. Then there exists a sub-model $\mathfrak{M}' = \langle \Omega', \Sigma',T', v'\rangle$ such that $w\in \Omega'$ and $\mathfrak{M}'\vDash \Gamma$.
\end{lemma}
\begin{proof}%[Proof of Lemma subL1]
	Let $\Gamma = \{\varphi_0, \varphi_1, \dots\}$. So we have  $\mathfrak{M}, w \vDash L_1^n \varphi_i$ for each $i , n \geq 0$.  Therefore,  $T(w, [\![L_1^{n}\varphi_i]\!]_\mathfrak{M})= 1$, for each $i,n\geq 0$. Put $\Omega_{i,n} = [\![L_1^{n}\varphi_i]\!]_\mathfrak{M}$ and $\Omega' = \bigcap_{i, n} \Omega_{i,n}$. Then, $w\in \Omega'$ and $\Omega'$ is a measurable set. Moreover,  for each $w'\in \Omega'$ and $i, n\geq  0$, we have  $\mathfrak{M}, w' \vDash L_1^{n+1}\varphi_i$. Hence, $T(w', \Omega_{i,n})= 1$, for each $i,n\geq 0$. This implies that
	$T(w',\Omega')=1$. Now, if we take $\mathfrak{M}' = \langle \Omega', \Sigma',T', v'\rangle$ as a sub-model of $\mathfrak{M} $, then $\mathfrak{M}'$ satisfies the required claim.
\end{proof}

\begin{definition}
	Let $\mathfrak{P}=\langle \Omega, \Sigma, T \rangle$ be a Markov process. We say that a Markov process $\mathfrak{P}'=\langle \Omega, \Sigma', T' \rangle$ is an \emph{event sub-process} of $\mathfrak{P}$ if $\Sigma'$ is a $\sigma$-sub-algebra of $\Sigma$ and $T' = T\upharpoonright_{\Sigma'}$ is a Markov kernel.
\end{definition}

The following lemma is immediate by Proposition~\ref{validity}.

\begin{lemma}
	Let $\mathfrak{P}'$ be an event sub-process of $\mathfrak{P}$. Then
	$\mathfrak{P}\models\varphi$ implies that $\mathfrak{P}'\models\varphi$,
	for any formula $\varphi$.
\end{lemma}

Next, we define disjoint unions of Markov processes.

\begin{definition}[Disjoint union]\label{diun}
	\begin{enumerate}
		\item Let $I$ be a non-empty index set, and for each $i \in I$ let
		$\mathfrak{P}_i = \langle \Omega_i, \Sigma_i, T_i \rangle$ be a Markov process.
		The \emph{disjoint union}
		$\biguplus_{i \in I} \mathfrak{P}_i$ is the Markov process
		$\mathfrak{P} = \langle \Omega, \Sigma, T \rangle$ defined as follows:
		\begin{itemize}
			\item $\Omega = \biguplus_{i \in I} \Omega_i$;
			\item $\Sigma$ is the $\sigma$-algebra generated by the algebra of all sets of the form
			$\biguplus_{i \in I} N_i$, where $N_i \in \Sigma_i$;
			\item For all $j\in I$ and $w \in \Omega_j$,
			\[
			T(w,\biguplus_{i \in I} N_i) = T_j(w,N_j).
			\]
		\end{itemize}
		
		\item If $\mathfrak{M}_i = \langle \Omega_i, \Sigma_i, T_i, v_i \rangle$ is a Markov
		model for each $i \in I$, then the disjoint union
		$\biguplus_{i \in I} \mathfrak{M}_i$ is the Markov model
		$\mathfrak{M} = \langle \biguplus_{i \in I} \mathfrak{P}_i, v \rangle$, where
		\[
		v(p) = \biguplus_{i \in I} v_i(p)
		\quad \text{for all } p \in \mathbb{P}.
		\]
	\end{enumerate}
\end{definition}

\begin{lemma}
	The disjoint union $\mathfrak{P} = \biguplus_{i \in I} \mathfrak{P}_i$ and $\mathfrak{M}=\biguplus_{i \in I}\mathfrak{M}_i$ are respectively a Markov process and a Markov model.
\end{lemma}

%\begin{proof}
%	Straightforward.
%\end{proof}

\begin{lemma}
	\begin{enumerate}
		\item Let $\{\mathfrak{M}_i\}_{i \in I}$ be a family of Markov models and let $\varphi$
		be a formula. Then for each $w \in \Omega_i$,
		\[
		\biguplus_{i \in I} \mathfrak{M}_i, w \vDash \varphi
		\quad \text{iff} \quad
		\mathfrak{M}_i, w \vDash \varphi.
		\]
		
		\item Let $\{\mathfrak{P}_i\}_{i \in I}$ be a family of Markov processes. Then for every
		formula $\varphi$,
		\[
		\biguplus_{i \in I} \mathfrak{P}_i \vDash \varphi
		\quad \text{iff} \quad
		\mathfrak{P}_i \vDash \varphi \text{ for all } i \in I.
		\]
	\end{enumerate}
\end{lemma}

Next, the notion of zigzag morphisms are presented that gives rise to an important notion of similarity between two Markov processes.

\begin{definition}[Zigzag morphisms]
	Let $\mathfrak{P} = \langle \Omega, \Sigma, T \rangle$ and $\mathfrak{P}' = \langle \Omega', \Sigma', T' \rangle$ be two Markov processes. A surjective measurable function $f: \Omega \to \Omega'$ is called a {\em zigzag  morphism} from $\mathfrak{P}$ to $\mathfrak{P}'$ if it satisfies the condition:
	\[T(w, f^{-1}(N')) = T'(f(w), N'),\] for each $w \in \Omega$ and $N' \in \Sigma'$.
\end{definition}

A zigzag morphism from $\mathfrak{M}  =\langle\mathfrak{P}, v \rangle$ to $\mathfrak{M}' = \langle\mathfrak{P}', v' \rangle$ is a zigzag  morphism $f$  from $\mathfrak{P}$ to $\mathfrak{P}'$  such that $v(p)= f^{-1} (v'(p))$, for each $p\in \mathbb{P}$.

Subsequently, the following preservation property can be derived.

\begin{lemma}\label{zigzag lem}
	Let $\mathfrak{P} = \langle \Omega, \Sigma, T \rangle$ and
	$\mathfrak{P}' = \langle \Omega', \Sigma', T' \rangle$ be two Markov processes.
	If there exists a zigzag morphism between $\mathfrak{P}$ and $\mathfrak{P}'$,
	then, $\mathfrak{P} \vDash \varphi \; \text{implies} \; \mathfrak{P}' \vDash \varphi$,
	for every formula $\varphi$.
\end{lemma}
\begin{proof}%[Proof of Lemma~\ref{zigzag lem}]
	Assume that $\mathfrak{P} \vDash \varphi$.
	Let $\mathfrak{M} = \langle \mathfrak{P}, v \rangle$ and
	$\mathfrak{M}' = \langle \mathfrak{P}', v' \rangle$ be arbitrary models,
	and let $f : \mathfrak{M} \to \mathfrak{M}'$ be a zigzag morphism.
	By the invariance of satisfaction under zigzag morphisms, for every world
	$w \in \Omega$ we have
	\[
	\mathfrak{M}, w \vDash \varphi
	\; \text{if and only if} \;
	\mathfrak{M}', f(w) \vDash \varphi .
	\]
	Since $\mathfrak{P} \vDash \varphi$, it follows that
	$\mathfrak{M}, w \vDash \varphi$ for all $w \in \Omega$, and hence
	$\mathfrak{M}', w' \vDash \varphi$ for all $w' \in \Omega'$.
	Therefore, $\mathfrak{P}' \vDash \varphi$.
\end{proof}
%%
%%
%%%%%%%%%%%%%%%%%%%%%%%%%%%%%%%%%%%%%%%%
%%%%%%%%%%%%%%%%%%%%%%%%%%%%%%%%%%%%%%%%
\subsection{Ultrafilter Extensions}
We now introduce ultrafilter extensions for countably generated Markov processes, a construction that parallels the classical ultrafilter extension of Kripke structures. Throughout this section, we assume familiarity with ultrafilters on Boolean algebras. %For a comprehensive account, see, for example, \cite{stone1937applications}.

We recall from \cite{kuter:model13} the definition of the \emph{ultrafilter limit} of a sequence $(r_w)_{w \in \Omega}$ of real numbers with respect to an ultrafilter $\mathfrak{u}$ over a Boolean algebra $\mathcal{B}$. This limit, denoted by $\lim_\mathfrak{u} r_w$, is defined as the unique real number $r \in [0,1]$ such that
\[
\{ w \in \Omega \mid |r_w - r| < \varepsilon \} \in \mathfrak{u}
\quad \text{for every } \varepsilon > 0.
\]

\begin{fact}\label{lim}
	Let $(r_w)_{w \in \Omega}$ be a sequence of real numbers in $[0,1]$, and let $\mathfrak{u}$ be an ultrafilter over $\Omega$. Then the following properties hold:
	\begin{enumerate}
		\item If $\{ w \in \Omega \mid r_w \geq r \} \in \mathfrak{u}$, then $\lim_\mathfrak{u} r_w \geq r$.
		\item If $\lim_\mathfrak{u} r_w > r$, then $\{ w \in \Omega \mid r_w > r \} \in \mathfrak{u}$.
		\item $\lim_\mathfrak{u} r_w \geq r$ if and only if $\{ w \in \Omega \mid r_w \geq s \} \in \mathfrak{u}$ for every $s < r$.
	\end{enumerate}
\end{fact}

\begin{definition}
	Let $\mathfrak{P} = \langle \Omega, \mathcal{B}, T \rangle$ be a countably generated Markov process.
	We define the structure $\langle \mathcal{U}^*, \mathcal{B}^*, K^* \rangle$ as follows:
	\begin{itemize}
		\item $\mathcal{U}^*$ is the set of all ultrafilters on $\mathcal{B}$;
		\item $\mathcal{B}^* := \{ A^* \mid A \in \mathcal{B} \}$, where
		\[
		A^* := \{ \mathfrak{u} \in \mathcal{U}^* \mid A \in  \mathfrak{u}\};
		\]
		\item $K^* : \mathcal{U}^* \times \mathcal{B}^* \to [0,1]$ is defined by
		\[
		K^*(\mathfrak{u}, A^*) := \lim_\mathfrak{u} T(w, A).
		\]
	\end{itemize}
\end{definition}

\begin{lemma}\label{lem basic ue}
	The triple $\langle \mathcal{U}^*, \mathcal{B}^*, K^* \rangle$ satisfies the following properties:
	\begin{enumerate}
		\item $\mathcal{B}^*$ is a countable Boolean algebra over $\mathcal{U}^*$. Moreover, $\mathcal{B}^*$ forms a basis of clopen sets for the compact Hausdorff Stone topology $\beta(\mathcal{U}^*)$;
		\item for each $\mathfrak{u} \in \mathcal{U}^*$, the function $K^*(\mathfrak{u}, \cdot) : \mathcal{B}^* \to [0,1]$ is a pre-measure.
	\end{enumerate}
\end{lemma}
\begin{proof}%[Proof of Lemma~\ref{lem basic ue}]
	\begin{enumerate}
		\item For $A,B \in \mathcal{B}$, the following equalities hold:
		\[
		A^* \cap B^* = (A \cap B)^*, \qquad
		A^* \cup B^* = (A \cup B)^*, \qquad
		(A^*)^c = (A^c)^*,
		\]
		since elements of $\mathcal{U}^*$ are ultrafilters. The remaining claims follow from standard results in Stone duality; see, for example, \cite{Johnstone1982}.
		\item It is straightforward to verify that $K^*(\mathfrak{u},\cdot)$ is finitely additive.
		Now suppose $(B_n^*)_{n \in \mathbb{N}}$ is a decreasing sequence in $\mathcal{B}^*$ such that
		\[
		\bigcap_{n=1}^\infty B_n^* = \varnothing.
		\]
		Since each $B_n^*$ is clopen in the compact space $\beta(\mathcal{U}^*)$, there exists $N \in \mathbb{N}$ such that $B_n^* = \varnothing$ for all $n > N$. Consequently,
		\[
		K^*(\mathfrak{u}, B_n^*) = 0 \quad \text{for all } n > N,
		\]
		and hence $\lim_{n \to \infty} K^*(\mathfrak{u}, B_n^*) = 0$.
	\end{enumerate}
\end{proof}

By Carathéodory’s Extension Theorem, the pre-measure $K^*$ extends uniquely to a Markov kernel
\[
T^* : \mathcal{U}^* \times \Sigma^* \to [0,1],
\]
where $\Sigma^*$ is the Borel $\sigma$-algebra generated by $\mathcal{B}^*$.

\begin{corollary}\label{P*}
	Let $\mathfrak{P} = \langle \Omega, \mathcal{B}, T \rangle$ be a countably generated Markov process.
	Then the structure
	\[
	\mathfrak{P}^* = \langle \mathcal{U}^*, \mathcal{B}^*, T^* \rangle
	\]
	is a countably generated Stone--Markov process.
\end{corollary}
\begin{proof}%[Proof of Corollary~\ref{P*}]]
	It suffices to verify that for each $N \in \Sigma^*$, the function
	\[
	T^*(\cdot, N) : \mathcal{U}^* \to [0,1]
	\]
	is measurable. That is, for every $r \in \mathbb{Q}_0$,
	\[
	\{ \mathfrak{u} \in \mathcal{U}^* \mid T^*(\mathfrak{u}, N) \geq r \} \in \Sigma^*.
	\]
	Without loss of generality, assume $N = B^*$ for some $B \in \mathcal{B}$. By Fact~\ref{lim}(3), we have
	\[
	F_r(B^*)
	= \{ \mathfrak{u} \in \mathcal{U}^* \mid T^*(\mathfrak{u}, B^*) \geq r \}
	= \bigcap_{s<r} (F_s(B))^*.
	\]
	Since $(F_s(B))^* \in \Sigma^*$ for all $s < r$, it follows that $F_r(B^*) \in \Sigma^*$.
\end{proof}

Next, we isolate the key notion of nested Archimedean ultrafilters, which constitute the basic building blocks of ultrafilter extensions.
%{\color{orange}Recall the definition of $F_r$ from Definition~\ref{coun gen def}. Let $F_{r_1\dots r_n}=F_{r_1}\dots F_{r_n}$.}\rhn{Or add it above, when $F_r$ is defined.}
\begin{definition}
	Let $\mathfrak{P} = \langle \Omega, \mathcal{B}, T \rangle$ be a countably generated Markov process.
	An ultrafilter $\mathfrak{u}$ on $\mathcal{B}$ is called a \emph{nested Archimedean ultrafilter} if, for all
	$t_1,\dots,t_n,r \in \mathbb{Q}_0$ and all $A \in \mathcal{B}$, the following condition holds:
	whenever
	\[
	F_{t_1 \dots t_n s}(A) \in \mathfrak{u} \quad \text{for all } s < r,
	\]
	then
	\[
	F_{t_1 \dots t_n r}(A) \in \mathfrak{u}.
	\]
\end{definition}

Let $\mathcal{U}^*_{na}$ denote the set of all nested Archimedean ultrafilters. We show that $\mathcal{U}^*_{na}$ forms a $G_{\delta}$-subspace of $\mathcal{U}^*$  such that $T^*(u, \mathcal{U}^*_{na}) = 1$, for each $\mathfrak{u}\in \mathcal{U}^*_{na}$.

\begin{lemma}\label{star}
	The following properties hold.
	\begin{enumerate}
		\item $\mathcal{U}^*_{na}$ forms a dense $G_{\delta}$-subspace of $\mathcal{U}^*$;
		\item $T^*(\mathfrak{u}, \mathcal{U}^*_{na}) = 1$, for each $\mathfrak{u}\in \mathcal{U}^*_{na}$;
		\item $\{\mathfrak{u}\in \mathcal{U}^*_{na}\;|\; T^*(\mathfrak{u}, B^*) \geq r \} = (F_r(B))^* \cap \mathcal{U}^*_{na}$, for each $B \in \mathcal{B}$.
	\end{enumerate}
\end{lemma}
\begin{proof}%[Proof of Lemma~\ref{star}]
	\begin{enumerate}
		\item It is shown in \cite{kozen2013stone} that
		\[\mathcal{U}^*_{na}= \bigcap_{B\in \mathcal{B}} \;\; \bigcap_{t_1, \dots, t_n, r\in \mathbb{Q}_0} \Big(\bigcup_{s<r}  \big(F^c_{t_1\cdots t_n s}  (B)\big)^* \;\cup\; \big(F_{t_1\cdots t_n r} (B)\big)^*\Big).\]
		Since  $\mathcal{O}_{t_1\cdots t_n r}(B):= \Big(\bigcup_{s<r}  \big(F^c_{t_1\cdots t_n s}  (B)\big)^* \cup \big(F_{t_1\cdots t_n r} (B)\big)^*\Big)$ is an open subset of $\mathcal{U}^*$, it follows that $\mathcal{U}^*_{na}$ is a $G_{\delta}$-subspace of $\mathcal{U}^*$. Moreover, by Rasiowa–Sikorski Lemma \cite{RasiowaSikorski1950},
		%\ref{good}\rhn{We never add it, and maybe we don't need to add it.},
        there exists an ultrafilter  $\mathfrak{u}\in \mathcal{U}^*$ such that for  $\emptyset \not= B'\in \mathcal{B} $ and $t_1, \dots, t_n, r\in \mathbb{Q}_0$, we have that $B'\in \mathfrak{u}$ and whenever
		\[F_{t_1\cdots t_n s}  (B)\in \mathfrak{u}\;\;\;\; \text{for all}\; s<r,\]
		then
		\[F_{t_1\cdots t_n r}  (B)\in \mathfrak{u}.\]
		Therefore, $\mathfrak{u}\in \mathcal{O}_{t_1\cdots t_n r}(B) \cap (B')^*$. This implies that $\mathcal{O}_{t_1\cdots t_n r}(B)$ is dense and open. But, Since $\mathcal{U}^*$ is topologically complete, $\mathcal{U}^*_{na}$ is dense.
		\item  It follows from Item 1 that $\mathcal{U}^*_{na}\in \Sigma^*$. Now to show the claim, it suffices to prove that for each $\mathfrak{u}\in\mathcal{U}^*_{na}\!$ ,
		\[T^*(\mathfrak{u}, \mathcal{U}^*\setminus\mathcal{U}^*_{na}) =0.\]
		Equivalently, we have to show that for each $B\in \mathcal{B}$ and $t_1, \dots, t_n, r\in \mathbb{Q}_0$,
		\begin{align*}
			&\;T^*(\mathfrak{u}, \bigcap_{s<r}  \big(F_{t_1\cdots t_n s}  (B)\big)^* \cap \big(F^c_{t_1\cdots t_n r} (B)\big)^* ) \\
			=& \;T^*(\mathfrak{u}, \bigcap_{s<r}  \big(F_{t_1\cdots t_n s}  (B)\big)^*) - T^*(\mathfrak{u}, \big(F_{t_1\cdots t_n r} (B)\big)^* )\\
			=& \; \inf_{s<r}{\{T^*(\mathfrak{u}, \big(F_{t_1\cdots t_n s}  (B)\big)^*) \}} - T^*(\mathfrak{u}, \big(F_{t_1\cdots t_n r} (B)\big)^* )\\
			=& \; 0
		\end{align*}
		
		Let $\beta := T^*(\mathfrak{u}, \big(F_{t_1\cdots t_n r} (B)\big)^* )$. Then,
		
		{\bf Claim:} For each positive $\epsilon \in \mathbb{Q}_0$, there exists $s<r$ such that
		\[T^*(\mathfrak{u}, \big(F_{t_1\cdots t_n s}  (B)\big)^* ) \leq \epsilon +\beta.\]
		{\it Proof of claim}: Suppose not. Then, for some $\epsilon$, we have
		\[ \lim_u T(w, F_{t_1\cdots t_n s}  (B)) > \epsilon +\beta,\]		
		for each $s< r$. Now pick $\epsilon' \in \mathbb{Q}_0$ such that $0 \leq \beta< \epsilon' \leq \epsilon +\beta$. Then, for each $s< r$,
		%	\[ \lim_\mathfrak{u} \{T(w, F_{t_1\cdots t_n s}  (B))\;|\; w\in \Omega\} \geq \epsilon'.\]	
		\[\{w\in \Omega\;|\; T(w, F_{t_1\cdots t_n s}(B))\geq \epsilon'\}\in \mathfrak{u}.\]
		Since $\mathfrak{u}$ is a nested Archimedean ultrafilter, it follows that \[\{w\in \Omega\;|\; T(w, F_{t_1\cdots t_n r}(B))\geq \epsilon'\}\in \mathfrak{u}.\]
		So,
		\[\lim_\mathfrak{u}  T(w, F_{t_1\cdots t_n r}  (B)) = \beta  \geq \epsilon',\]
		a contradiction.
		
		\item This holds, since $(F_r(B))^*\cap \mathcal{U}^*_{na} = \bigcap_{s<r} ((F_s(B))^* \cap \mathcal{U}^*_{ua})$.
	\end{enumerate}
\end{proof}

In light of the above lemma, we define the ultrafilter extensions.

\begin{definition}[Ultrafilter extensions] \label{ultrafil}
	Given a Markov process $\mathfrak{P} = \langle \Omega, \mathcal{B}, T \rangle$, we define the {\em ultrafilter extension} of $\mathfrak{P}$ as $\mathfrak{Ue(P)} = \langle \mathcal{U}^*_{na}, \Sigma^*_{na}, T^*_{na}\rangle$ where
	\begin{enumerate}
		%\item $\mathcal{U}^*_{na}$ is the set of all good ultrafilters over $\Omega$;
		\item $\Sigma^*_{na} := \{ N\cap\mathcal{U}^*_{na} \;|\;N\in \Sigma^*\}$;
		\item $T^*_{na}:\mathcal{U}^*_{na}\times \Sigma^*_{na} \to [0, 1]$ such that for each $\mathfrak{u}\in \mathcal{U}^*_{na}$ and $N\in \Sigma^*$,
		\[T^*_{na} (\mathfrak{u}, N\cap\mathcal{U}^*_{na}) = T^*(\mathfrak{u}, N).\]
	\end{enumerate}
\end{definition}

Equivalently, $\mathfrak{Ue(P)}$ is a sub-process of $\mathfrak{P}^* = \langle \mathcal{U}^*, \mathcal{B}^*, T^* \rangle$, introduced in Corollary \ref{P*}. Notice that $\mathfrak{Ue(P)} = \langle \mathcal{U}^*_{na}, \mathcal{B}^*_{na}, T^*_{na} \rangle$ where $\mathcal{B}^*_{na}:= \{B^* \cap\mathcal{U}^*_{na}\;|\; B\in\mathcal{B} \}$.

In the following lemma, we show that the validity of formulas is {\em reflected} by the ultrafilter extensions. To prove this, we first examine the following lemma.

\begin{lemma}\label{lem ue}
	Let $\mathfrak{M} = \langle \Omega, \mathcal{B}, T, v \rangle$ be a model such that $\mathcal{B}$ is countable. Consider the model $\mathfrak{Ue(M)} = \langle \mathfrak{Ue(P)}, v^*_{na} \rangle$ where $\mathfrak{Ue(P)} =  \langle\mathcal{U}^*_{na}, \mathcal{B}^*_{na}, T^*_{na} \rangle$
	and  $v^*_{na}$ is a valuation on $\mathfrak{Ue(P)}$ defined as \[v^*_{na}(p):= 
    \{ \mathfrak{u}\in\mathcal{U}^*_{na}\;|\; v(p) \in \mathfrak{u} \},\;\;\; \text{for each $p\in \mathbb{P}$}.\]
	Then, the following conditions hold:
	\begin{itemize}
		\item[1.] $v^*_{na}(p) = \mathcal{U}^*_{na} \cap (v(p))^*\in\mathcal{B}^*_{na}$, for each $p\in \mathbb{P}$.
		\item[2.] For each  formula $\varphi$ and $\mathfrak{u}\in \mathcal{U}^*_{na}$, we have  
		\[\mathfrak{Ue(M)}, \mathfrak{u}\vDash \varphi \;\;\text{if and only if}\;\; \{w\in \Omega\;|\; \mathfrak{M}, w\vDash \varphi\}\in \mathfrak{u}.\]
		Therefore,  $[\![\varphi]\!]_{\mathfrak{Ue(M)}} = \mathcal{U}^*_{na}   \cap ([\![\varphi]\!]_{\mathfrak{M}})^* $. 
	\end{itemize}
\end{lemma}
\begin{proof}%[Proof of Lemma~\ref{lem ue}]
	\begin{enumerate}
		\item This follows immediately by definition. 
		\item This can be shown by induction on the complexity of formulas.
		The basic and induction steps for Boolean connectives are clear.
		We assume that the induction hypothesis holds for $\psi$. Now for $\varphi= L_r \psi$, we have
		\begin{align*}
			[\![ L_r \psi]\!]_{\mathfrak{Ue(M)}} &= \{\mathfrak{u}\in\mathcal{U}^*_{na} \;|\; T^*_{na}(\mathfrak{u}, [\![\psi]\!]_{\mathfrak{Ue(M)}})\geq r\}\\
			&= \{\mathfrak{u}\in \mathcal{U}^*_{na} \;|\; T^*(\mathfrak{u}, \;\mathcal{U}^*_{na} \cap  ([\![\psi]\!]_{\mathfrak{M}})^*)\geq r\},  \;(\text{by induction hypothesis for $\psi$})\\
			&= \{\mathfrak{u}\in \mathcal{U}^*_{na} \;|\; T^*(\mathfrak{u}, \;([\![\psi]\!]_{\mathfrak{M}})^*)\geq r\}, \;(\text{by Definition \ref{ultrafil} (2))}\\
			&= \mathcal{U}^*_{na} \cap ([\![L_r \psi]\!]_{\mathfrak{M}})^*, \;(\text{by Lemma \ref{star} (3))}\\
			&= \mathcal{U}^*_{na} \cap ([\![L_r \psi]\!]_{\mathfrak{M}})^*.
		\end{align*}
	\end{enumerate}
\end{proof}

\begin{lemma}\label{ue reflect}
	Suppose $\mathfrak{P} = \langle \Omega, \Sigma, T \rangle$ is a countably generated Markov processes. Then, for each formula  $\varphi$,
	\[\mathfrak{Ue(P)}\vDash \varphi \;\;\text{implies} \;\;\mathfrak{P}\vDash \varphi.\]
\end{lemma}
\begin{proof}%[Proof of Lemma~\ref{ue reflect}]
	Assume the above claim is not true for some formula $\varphi$ . Then by Proposition \ref{validity} (2), there exists a countable subalgebra $\mathcal{B}$ of $\Sigma$ such that  $\langle \Omega, \mathcal{B}, T \rangle\nvDash \varphi$. Hence, for some model $\mathfrak{M}= \langle \mathfrak{P}, v \rangle$  and $w\in \Omega$, we have that  $\mathfrak{M}, w \nvDash \varphi$.  So $ \emptyset \neq [\![\neg\varphi]\!]_{\mathfrak{M}} \in \mathcal{B}$. Now consider an ultrafilter $\mathfrak{u}\in \mathcal{U}^*_{na}$ with $[\![\neg\varphi]\!]_{\mathfrak{M}} \in u$, which exists by Rasiowa–Sikorski Lemma. Then on the basis of Lemma \ref{lem ue} (2), we have $\mathfrak{Ue(M)}, u \nvDash \varphi$, a contradiction.
\end{proof}

The proof of the above lemma easily yields the following corollary.

\begin{corollary}\label{ue preserve}
	Under the condition of the above corollary,  for each formula $\varphi$,
	\begin{center}
		$\mathfrak{P}\vDash \varphi$ implies $\langle \mathcal{U}^*_{na}, \mathcal{B}^*_{na}, T^*_{na}\rangle \vDash \varphi.$
	\end{center} 
\end{corollary}

We can conclude this section by the following lemma which shows that for a given theory $\Gamma$, the class of all Markov processes in which $\Gamma$ is valid is closed under the constructions introduced above.

\begin{lemma}\label{lem closed}
	Assume that $\Gamma$ is a set of formulas and let $\mathcal{C}_\Gamma=\{\mathfrak{P} \mid \mathfrak{P}\models\Gamma\}$, then $\mathcal{C}_\Gamma$ is closed under Markov sub-processes, event sub-processes, disjoint unions, zigzag morphisms and reflected by ultrafilter extensions of countably generated Markov processes\footnote{i.e., if $\mathfrak{Ue(P)}\in\mathcal{C}_\Gamma$, then $\mathfrak{P}\in\mathcal{C}_\Gamma$.}.
	Furthermore, $\mathfrak{P}\in\mathcal{C}$ if and only if every countably generated event sub-process of $\mathfrak{P}$ is in $\mathcal{C}$.
\end{lemma}

Since the conditions stated in the above lemma are important conditions for proving the Goldblatt-Thomason theorem, we isolate those classes that satisfy the conditions of the above lemma. 

\begin{definition}
	Let $\mathcal{C}$ be a class of Markov processes. The class $\mathcal{C}$ is said to have the Goldblatt-Thomason property (GT-property for short) if it satisfies the conditions of Lemma \ref{lem closed}.
\end{definition}
%%
%%
%%%%%%%%%%%%%%%%%%%%%%%%%%%%%%%%%%%%%%%%
%%%%%%%%%%%%%%%%%%%%%%%%%%%%%%%%%%%%%%%%
\section{The Goldblatt-Thomason Theorem}\label{sec gt}
In this section we prove the Goldblatt-Thomason theorem for probability logic.
Motivated by the same theorem in the context of basic modal logic with Kripke semantics we consider the ultrafilter extension as replacement of ultraproduct construction. This also rooted back to the original paper of Goldblatt and Thomason \cite{gt:axiom75}.

Recall that a class of Markov processes $\mathcal{C}$ is closed under ultrafilter extensions, if for every countably generated Markov process $\mathfrak{P}$, if $\mathfrak{P}\in\mathcal{C}$, then $\mathfrak{Ue(P)}\in\mathcal{C}$.

\begin{theorem}[Goldblatt-Thomason Theorem]
	Let $\mathcal{C}$ be a class of Markov processes that is closed under taking ultrafilter extensions. Then $\mathcal{C}$ is P-definable if and only if it has the GT-property.
\end{theorem}
\begin{proof}
	The implication from left to right follows directly from Lemma~\ref{lem closed}.
	For the converse direction, let $Th(\mathcal{C})$ denote the set of all formulas of an arbitrary countable language $\mathcal{L}$ that are valid in the class $\mathcal{C}$.
	We show that $Th(\mathcal{C})$ defines $\mathcal{C}$, i.e., for any Markov processes $\mathfrak{P} =\langle \Omega, \Sigma,T\rangle $, we have 
	\[\mathfrak{P}\in \mathcal{C}\;\;\text{iff}\; \;\mathfrak{P} \vDash Th(\mathcal{C}).\] 
	
	By definition of $Th(\mathcal{C})$, the forward implication is immediate.  Hence, assume that $\mathfrak{P} \vDash Th(\mathcal{C})$. 
	By Proposition \ref{validity}, we may assume without loss of generality that $\mathfrak{P} =\langle \Omega, \mathcal{B},T\rangle $, where $\mathcal{B}$ is a countable Boolean algebra and $\Sigma=\sigma(\mathcal{B})$.
	
	Now consider the language $\mathcal{L}_{\mathfrak{P}}$ consisting of propositional variables  $\{p_A\;|\; A\in \mathcal{B}\}$.  Define a valuation $v_{\mathfrak{P}}$  on  $\mathfrak{P}$ by setting  $v_{\mathfrak{P}}(p_A) = A$, for each $A\in \mathcal{B}$.
	Let $\mathfrak{M} = \langle \mathfrak{P}, v_{\mathfrak{P}}\rangle$ and let $\Gamma = Th(\mathfrak{M})$ be the set of formulas valid in $\mathfrak{M}$. Observe that $\Gamma$ contains the following formulas
	\begin{align}
		\neg p_{\emptyset} \\
		p_A\wedge p_B &\leftrightarrow p_{A\cap B}\\
		p_A &\leftrightarrow \neg p_{A^c}\\
		p_{F_rA} &\leftrightarrow L_r p_A
		%	p_{M_rA} &\leftrightarrow M_r p_A
	\end{align}
	for each $A, B \in  \mathcal{B}$ and $r\in \mathbb{Q}_0$. Moreover, if $\varphi\in \Gamma$ then $L_1\varphi\in \Gamma$.
	
	Now for each $w\in \Omega$, define $\Gamma_w = \{ \gamma\in \mathcal{L}_{\mathfrak{P}} \;|\; \mathfrak{M}, w\vDash \gamma\}$. Clearly, $Th(\mathcal{C})\subseteq \Gamma\subseteq \Gamma_w$.
	
	\medskip
	\noindent
	{\bf Claim 1:} $\Gamma_w$  is finitely satisfiable in $\mathcal{C}$.
	
	\medskip
	\noindent
	Suppose otherwise. Then there exists a finite subset $\Delta \subseteq \Gamma_w$ that is not satisfiable in $\mathcal{C}$. By definition of $Th(\mathcal{C})$, this implies 
	\[\neg \bigwedge \Delta \in Th(\mathcal{C}) \subseteq \Gamma_w,\] 
	contradicting $\Delta\subseteq \Gamma_w$.

	For each finite $\Delta\subseteq \Gamma_w$, let $\mathfrak{N}_\Delta= \langle \mathfrak{Q}_\Delta, v_\Delta\rangle$ be a model satisfying $\Delta$ with $\mathfrak{Q}_\Delta\in \mathcal{C}$. Define $ \mathfrak{N}'_w = (\mathfrak{Q}'_w, v_w)$ %\biguplus_{\Delta \subseteq \Gamma_w}  \mathfrak{M}_\Delta$ 
	where $\mathfrak{Q}'_w=\biguplus_{\Delta \subseteq \Gamma_w}\mathfrak{Q}_\Delta\in\mathcal{C}$ 
	and $v_w=\biguplus_{\Delta \subseteq \Gamma_w}v_\Delta$.
	Let $\mathcal{B}_w = \{[\![\varphi]\!]_{\mathfrak{N}_w}|\; \varphi \in  \mathcal{L}_{\mathfrak{P}}\}$, and define
	$\mathfrak{N}_w =\mathfrak{N}'_w\upharpoonright_{\mathcal{B}_w}$. Then $\mathfrak{N}_w $ is an event sub-process  based on $\mathfrak{Q}_w=(\mathfrak{Q}'_w)\upharpoonright_{\mathcal{B}_w}\in\mathcal{C}$.
	One can verify that $\mathfrak{u}_w= \{[\![\varphi]\!]_{\mathfrak{N}_w}|\; \varphi \in \Gamma_w\}$  is a nested Archimedean ultrafilter over $\mathcal{B}_w$.
	Hence, $\mathfrak{u}_w \in \mathfrak{Ue}(\mathfrak{N}_w)$ and $\mathfrak{Ue}(\mathfrak{Q}_w)\in \mathcal{C}$, since $\mathcal{C}$ is closed under taking ultrafilter extensions.
	Moreover, by Lemma \ref{lem ue}, $ \mathfrak{Ue}(\mathfrak{N}_w), \mathfrak{u}_w\vDash  \Gamma_w$.
	Applying Lemma \ref{subL1}  to  $\Gamma$, we obtain a sub-model $\mathfrak{M}'_w =  \langle \mathfrak{P}'_w,v'_w\rangle$
	with $\mathfrak{P}'_w \in \mathcal{C}$ and $\mathfrak{u}_w\in \mathfrak{M}'_w$ such that $\mathfrak{M}'_w\vDash \Gamma$ and $\mathfrak{M}'_w, \mathfrak{u}_w\vDash\Gamma_w$.
	Define 
	$$\mathfrak{M}' =  (\biguplus_{w \in \Omega} \mathfrak{M}_w)\upharpoonright_{\mathcal{B}'} = \langle \mathfrak{P}',v'\rangle,$$ 
	and $\mathfrak{P}' = \langle \Omega', \mathcal{B}', T',v'\rangle$, where $\mathcal{B}' = \{[\![\varphi]\!]_{\biguplus_{w \in \Omega}\mathfrak{N}_w}|\; \varphi \in  \mathcal{L}_{\mathfrak{P}}\}$. Thus, $ \mathfrak{Ue}(\mathfrak{P'})\in \mathcal{C}$.
	
	\medskip
	\noindent
	{\bf Claim 2:} There is a zigzag morphism from $\mathfrak{Ue}(\mathfrak{P}')$ to $\mathfrak{Ue}(\mathfrak{P})$.
	
	\medskip
	\noindent	
	Define $f: \mathfrak{Ue}(\mathfrak{P}') \to \mathfrak{Ue}(\mathfrak{P})$ by $f(\mathfrak{u}') = \{A \in \mathcal {B}\;|\; \mathfrak{Ue}(\mathfrak{M}'), \mathfrak{u}' \vDash p_A\}$, for each $\mathfrak{u}' \in \Omega'$.
	\begin{enumerate}
		\item $Im(f)\subseteq\mathfrak{Ue}(\mathfrak{P})$:  By the definition of $\Gamma$ and this fact that $p_A\to p_B \equiv \top$ for each $A\subseteq B \in \mathcal{B}$, it is easy to see that $f(\mathfrak{u}')$ is an ultrafilter over $\mathcal{B}$, for each $\mathfrak{u}' \in \Omega'$. To prove that it is a nested Archimedean ultrafilter, assume that $$F_{t_1\dots t_n s} A \in f(\mathfrak{u}') \; \; \; \; \; \forall s<r.$$ 
		Then, 
		$$\mathfrak{Ue}(\mathfrak{M}), \mathfrak{u}' \vDash L_{t_1\dots t_n s} p_A \; \; \; \; \; \forall s<r.$$
		So we have $\mathfrak{Ue}(\mathfrak{M}), \mathfrak{u}' \vDash L_{t_1\dots t_n r} p_A$ and by applying axiom 4 of $\Gamma$, $\mathfrak{Ue}(\mathfrak{M}), \mathfrak{u}' \vDash p_{F_{t_1\dots t_n r} A}$.
		So $F_{t_1\dots t_n r} A\in f(\mathfrak{u}')$.
		\item $f$ is surjective: Let $\mathfrak{u}$ be a nested Archimedean ultrafilter over $\mathcal{B}$. Then, consider $\mathfrak{v}:= \{[\![p_A]\!]_{\mathfrak{M}'}\;|\; A\in \mathfrak{u}\}$. We show that $\mathfrak{v}$ is a nested Archimedean ultrafilter on $\mathcal{B}'$. Since all formulas of $\Gamma$ are valid in $\mathfrak{M}'$, it follows that every formula $\varphi$ in the language $\mathcal{L}_{\mathfrak{P}}$ is logically equivalent to a propositional variable $p_A$. Now since $\mathfrak{u}$ is an ultrafilter over $\mathcal{B}$, either  $[\![p_A]\!]_{\mathfrak{M}'}\in  \mathfrak{v} $ or $[\![p_A]\!]^c_{\mathfrak{M}'}\in  \mathfrak{v}$.
		Also, if $[\![p_A]\!]_{\mathfrak{M}'} , [\![p_B]\!]_{\mathfrak{M}'} \in \mathfrak{v}$, then so  $[\![p_A]\!]_{\mathfrak{M}'} \cap [\![p_B]\!]_{\mathfrak{M}'} = [\![p_{A\cap B}]\!]_{\mathfrak{M}'} \in \mathfrak{v}$.
		Moreover, since $A\in \mathfrak{u}$ is non-empty,  there exists $w\in A$. Therefore, $[\![p_A]\!]_{\mathfrak{M}'} \in \mathfrak{u}_w = \{[\![\varphi]\!]_{\mathfrak{N}_w}|\; \varphi \in \Gamma_w\}$. This implies $[\![p_A]\!]_{\mathfrak{M}'}\neq\emptyset$.
		Hence $\emptyset\not\in\mathfrak{v}$. 
		Finally, to show that $\mathfrak{v}$  is a nested Archimedean ultrafilter, let $A\in \mathfrak{u}$ and 
		$$F_{t_1\dots t_ns} [\![p_A]\!]_{\mathfrak{M}'}\in\mathfrak{v} \; \; \; \; \;  \forall s<r.$$ 
		Then  
		$$F_{t_1\dots t_ns} [\![p_A]\!]_{\mathfrak{M}'} =  [\![p_{F_{t_1\dots t_ns}  A}]\!]_{\mathfrak{M}'} \in \mathfrak{v}  \; \; \; \; \;  \forall s<r.$$ 
		Therefore, 
		$$F_{t_1\dots t_ns}  A\in \mathfrak{u}  \; \; \; \; \;   \forall s<r.$$ 
		But since $\mathfrak{u}$ is a nested Archimedean ultrafilter, it follows that $F_{t_1\dots t_nr}  A\in \mathfrak{u}$. So,  $[\![p_{F_{t_1\dots t_nr}  A}]\!]_{\mathfrak{M}'}   = F_{t_1\dots t_nr}  [\![p_{ A}]\!]_{\mathfrak{M}'} \in  \mathfrak{v}$.
		\item $f$ is a measurable zigzag morphism:  We have to verify that for $A^* = \{\mathfrak{u}\in \mathcal{U}^*_{na}\;|\; A\in \mathfrak{u}\}$,   $ f^{-1}(A^*)$ is measurable and
		\[T'^*_{na}(\mathfrak{u}, f^{-1}(A^*)) = T^*_{na}(f(\mathfrak{u}), A^*) \;\;\;\;(\star) \]
		for each $\mathfrak{u}\in \Omega'$. To see the measurability of $f$,  notice that
		\[f^{-1}(A^*)= \{\mathfrak{u}'\in \mathcal{U}'^*_{na} \; | \;  [\![p_A]\!]_{\mathfrak{M}'} \in \mathfrak{u}'\} \in \mathcal{B}^*_{na}.\]
		On the other hand, to show $(\star)$, assume  for some $r\in \mathbb{Q}_0$, we have \[T'^*_{na}(\mathfrak{u}, f^{-1}(A^*)) = \lim_{\mathfrak{u}'} T'(w', [\![p_A]\!]_{\mathfrak{M}'}) \geq r.  \]  
		As $\mathfrak{u}'$ is a nested Archimedean ultrafilter and $[\![p_{F_r A}]\!]_{\mathfrak{M}'}  = [\![L_r p_{A}]\!]_{\mathfrak{M}'} $, we have $[\![p_{F_r A}]\!]_{\mathfrak{M}'}  \in \mathfrak{u}'$.
		Therefore, $F_r A \in f(\mathfrak{u}')$. So $\mathfrak{Ue}(\mathfrak{M}), \mathfrak{u}'\vDash L_r p_A$ which implies that \[T_{\mathcal{G}}(f(\mathfrak{u}'), [\![p_A]\!]_{\mathfrak{Ue}(\mathfrak{M})}) = T_{\mathcal{G}}(f(\mathfrak{u}'), A^*)\geq r.\]
		Conversely, one can show that if $T^*_{na}(f(\mathfrak{u}'), A^*)\geq r$, then so $ T'^*_{na}(\mathfrak{u}', f^{-1}(A^*))\geq r$. Thus, $(\star)$ is achieved. 		
	\end{enumerate}	
	Since $f$ is a zigzag morphism, $\mathfrak{Ue(P)}\in\mathcal{C}$. 
	Moreover, as $\mathcal{C}$ reflects ultrafilter extensions, it follows that $\mathfrak{P}\in\mathcal{C}$.
\end{proof}	
%%
%%
%%%%%%%%%%%%%%%%%%%%%%%%%%%%%%%%%%%%%%%%
%%%%%%%%%%%%%%%%%%%%%%%%%%%%%%%%%%%%%%%%
\subsection{Some Examples}
%\subsection{Harsanyi type Markov processes}

Here in this subsection, we show that the class of Harsanyi type spaces \cite{harsan:games68} is an important instance for a P-definable class which is closed under ultrafilter extensions, and has the GT-property.

\begin{example}\label{harsanyi exm}
	For a Markov process $\mathfrak{P}=\langle{\Omega, \Sigma, T}\rangle$ and $w\in \Omega$ let 
	$$[T(w)] = \{w'\in \Omega \mid T(w)=T(w')\}.$$
	Recall that a Markov process $\mathfrak{P}$ is of {\em Harsanyi type} if
	$T(w,E)=1$, for each $w\in\Omega$ and $E\in\Sigma$ with $[T(w)]\subseteq E$.
	Note that $[T(w)]$ is not necessarily in $\Sigma$, unless $\Sigma$ is countably generated \cite[Page 10]{zhou:hars14}.
	In fact in case of countably generated Markov process $\mathfrak{P}=\langle\Omega,\mathcal{B},T\rangle$, 
	\begin{align}\label{t(w)}
		[T(w)] & = \bigcap_{B\in \mathcal{B}} \{w'\in \Omega\;|\; T(w)(B) = T(w')(B) \}\\
		&= \bigcap_{B\in \mathcal{B}} \bigcap_{r\in \mathbb{Q}_0} \{w'\in \Omega\;|\; T(w)(B)\geq r\leftrightarrow   T(w')(B)\geq r \}.
	\end{align}
	We denote the class of all Harsanyi type spaces by $\mathcal{H}ar$.
\end{example}
%%
%%	

%	\begin{proposition}
	%		Let $\mathfrak{P}$ be a Markov process. Then $\mathfrak{P}\in\mathcal{H}ar$ if and only if every countably generated event sub-process $\mathfrak{P}'$ of $\mathfrak{P}$ is in $\mathcal{H}ar$.
	%	\end{proposition}
%	\begin{proof}
	%		Let $\mathfrak{P}'= \langle\Omega,\Sigma',T'\rangle$ with $T'=T\upharpoonright_{\Sigma'}$ be a countably generated event sub-process of $\mathfrak{P}$. Then for each $w\in\Omega$ we have
	%		$[T'(w)]\in\Sigma'\subseteq\Sigma$ and 
	%		$[T(w)]\subseteq [T'(w)]$.
	%		Moreover, $T'(w,E) = T(w,E)$, for each $E\in\Sigma'$. 
	%		Therefore, if $\mathfrak{P}\in\mathcal{H}ar$, then so $\mathfrak{P}'\in\mathcal{H}ar$.
	%		Now assume that $\mathfrak{P}$ is not Harsanyi.
	%		So, $T(w,E)<1$ for some $w\in\Omega$ and some $E\in\Sigma$ with $[T(w)]\subseteq E$.??????
	%	\end{proof}

\begin{proposition}\label{har def}
	Assume that $\mathfrak{P}$ is a countably generated Harsanyi type space.
	Then $\mathfrak{P}\in\mathcal{H}ar$ if and only if the following set of formulas is valid in $\mathfrak{P}$:
	$$\Gamma_{har}= \{L_rp\to L_1L_rp \mid r\in\mathbb{Q}_0 \} \cup \{\neg L_rp\to L_1\neg L_rp \mid r\in\mathbb{Q}_0 \}.$$
\end{proposition}
\begin{proof}%[Proof of Proposition~\ref{har def}]
	It is known that,\cite{zhou:hars14}, these formulas are valid on any Harsanyi type space.
	
	%By the above mentioned fact, it is clear that $HAR$ is valid in the class (countably generated) $\mathcal{H}ar$. 
	Now assume that $\mathfrak{P}=\langle{\Omega,\Sigma, T}\rangle$ is a countably generated Markov process which is not Harsanyi type. Hence, there is $w\in\Omega$ such that $T(w)([T(w)])<1$. Therefore, $T(w)([T(w)]^c)>0$. Let $\Sigma=\sigma(\mathcal{B})$ where $\mathcal{B}$ is countable. By \ref{t(w)}, we have
	$$[T(w)]^c = \bigcup_{B\in \mathcal{B}} \bigcup_{r\in \mathbb{Q}_0} \{w'\in \Omega\;|\; T(w)(B)\geq r\leftrightarrow  T(w')(B)\geq r \}^c.$$ Thus, there are  $B\in \mathcal{B}$ and $r \in \mathbb{Q}_0$ such that either $T(w)(B) \geq r$, while $T(w)((F_r (B))^c) >0$ or $T(w)(B)< r$, while $T(w)(F_r (B)) >0$.
	In either of these cases, if we define a valuation $v$ over $\mathfrak{P}$ as $v(p) = B$, then we turn $\mathfrak{P}$ into a model $\mathfrak{M}$ such that we have either \[ \mathfrak{M},w\models L_rp\wedge \neg L_1L_rp \;\;\;\;\text{or} \;\;\;\; \mathfrak{M},w\models \neg L_rp\wedge  L_1 \neg L_rp.\]
\end{proof}

%{\color{orange} Find an example of non-countably generated space which is not Harsanyi but $HAR$ is valid in it! Or improve the above example to the class of all spaces.}

\begin{proposition}\label{har closed ue}
	The class of all countably generated Harsanyi type spaces is closed under ultrafilter extensions.
\end{proposition}
\begin{proof}%[Proof of Proposition~\ref{har closed ue}]
	Assume that $\mathfrak{P}=\langle\Omega,\mathcal{B},T\rangle$ is a Harsanyi type space with $\mathcal{B}$ countable.
	%	By Proposition~\ref{har def} and Corollary~\ref{ue preserve}, we have that  $\langle \mathcal{U}^*_{na}, \mathcal{B}^*_{na}, T^*_{na}\rangle \vDash\Gamma_{har}$.
	Let $\mathfrak{u}\in\mathcal{U}^*_{na}$. Then
	\begin{align}
		[T_{na}^*(\mathfrak{u})] & = \bigcap_{B^*\in \mathcal{B}^*_{na}} \{\mathfrak{u}'\in \mathcal{U}^*_{na}\;|\; T^*_{na}(\mathfrak{u}, B^*) = T^*_{na}(\mathfrak{u}' ,B^*) \}\\
		&= \bigcap_{B^*\in \mathcal{B}^*_{na}} \bigcap_{r\in \mathbb{Q}_0} \{\mathfrak{u}'\in \mathcal{U}^*_{na}\;|\; T^*_{na}(\mathfrak{u} ,B^*)\geq r\leftrightarrow   T^*_{na}(\mathfrak{u}',B^*)\geq r \}.
	\end{align}
	Now for $B\in\mathcal{B}$ and $r\in\mathbb{Q}_0$ there are two cases to consider.
	
	Case 1: $T^*_{na}(\mathfrak{u},B^*)< r$. In this case 
	\begin{align*}
		\{\mathfrak{u}'\in \mathcal{U}^*_{na}\;|\; T^*_{na}(\mathfrak{u} ,B^*)\geq r\leftrightarrow   T^*_{na}(\mathfrak{u}',B^*)\geq r \}  & =\{\mathfrak{u}'\in \mathcal{U}^*_{na}\;|\; T^*_{na}(\mathfrak{u}',B^*)< r \}\\
		& =  \{\mathfrak{u}'\in \mathcal{U}^*_{na}\;|\; (F_r(B))^c \in\mathfrak{u'}\}.
	\end{align*}
	So
	\begin{align*}
		T^*_{na}(\mathfrak{u}, \{\mathfrak{u}'\in \mathcal{U}^*_{na}\;|\; T^*_{na}(\mathfrak{u}',B^*)< r \}) &= \lim_{\mathfrak{u}} T(w,(F_r(B))^c)
	\end{align*}
	On the other hand 
	$T^*_{na}(\mathfrak{u},B^*) =\lim_{\mathfrak{u}} T(w,B)<r$.
	Therefore $\{w\in\Omega \; \mid\; T(w,B)<r\}\in\mathfrak{u}$.
	Notice that since $\mathfrak{P}$ is Harsanyi,
	$$\{w\in\Omega \mid T(w,B)<r\} \subseteq \{w\in\Omega \mid T(w,(F_r(B))^c)=1\}\in\mathfrak{u}.$$
	Thus $\lim_{\mathfrak{u}} T(w,(F_r(B))^c) =1$.
	
	Case 2: $T^*_{na}(\mathfrak{u},B^*)\geq r$.
	In this case we have 
	\begin{align*}
		\{\mathfrak{u}'\in \mathcal{U}^*_{na}\;|\; T^*_{na}(\mathfrak{u} ,B^*)\geq r\leftrightarrow   T^*_{na}(\mathfrak{u}',B^*)\geq r \}  & =\{\mathfrak{u}'\in \mathcal{U}^*_{na}\;|\; T^*_{na}(\mathfrak{u}',B^*)\geq r \}\\
		& =  \{\mathfrak{u}'\in \mathcal{U}^*_{na}\;|\; F_s(B) \in\mathfrak{u'}, \; \forall s<r\}\\
		& =  \{\mathfrak{u}'\in \mathcal{U}^*_{na}\;|\; F_r(B) \in\mathfrak{u'}\}.
	\end{align*}
	We also have $T^*_{na}(\mathfrak{u},B^*) =\lim_{\mathfrak{u}}T(w,B) >s$, for each $s<r$.
	So, 
	$$\{w\in\Omega \mid T(w,B)> s\}\subseteq\{w\in\Omega \mid T(w,B)\geq s\}\in\mathfrak{u}, \; \; \forall s<r.$$
	Hence $\{w\in\Omega \mid T(w,B)\geq r\} \in\mathfrak{u}$, since $\mathfrak{u}$ is a nested Archimedean ultrfilter.
	Again this yields that $\{w\in\Omega \mid T(w,F_r(B)) =1 \}\in\mathfrak{u}$, and
	$\lim_{\mathfrak{u}} T(w,F_r(B)) =1$.
	The above argument shows that 
	$T^*_{na}(\mathfrak{u}, (F_r(B))^*)=1$.
	
	From both cases we conclude  that 
	$$T^*_{na}(\mathfrak{u},\{\mathfrak{u}'\in \mathcal{U}^*_{na}\;|\; T^*_{na}(\mathfrak{u} ,B^*)\geq r\leftrightarrow   T^*_{na}(\mathfrak{u}',B^*)\geq r \})=1,$$ for each $B\in\mathcal{B}$ and $r\in\mathbb{Q}_0$.
	Thus $T^*_{na}(\mathfrak{u},[T^*_{na}(\mathfrak{u})])=1$.
\end{proof}

\begin{proposition}\label{har gt}
	The class $\mathcal{H}ar$ has the GT-property. 
	%of Harsanyi type Markov process is closed under the taking event sub-processes, zigzag morphisms and disjoint unions. Furthermore, the class of countably generated Harsanyi type Markov process is closed under the ultrafilter extensions.
\end{proposition}
\begin{proof}%[Proof of Proposition ref{har gt}]
	%The proof follows from Lemma \ref{lem closed}, Corollary~\ref{ue reflect} and Example~\ref{harsanyi exm}.
	We only show that $\mathcal{H}ar$ reflects the ultrafilter extensions, since the other properties can be easily shown.
	Assume that $\mathfrak{Ue(P)}=\langle \mathcal{U}^*_{na}, \Sigma^*,T^*_{na}\rangle \in \mathcal{H}ar$ for a countably generated $\mathfrak{P}=\langle\Omega,\Sigma,T\rangle$.
	Then $\mathfrak{Ue(P)}\models\Gamma_{har}$. So, by Lemma \ref{ue reflect}, we have $\mathfrak{P}\models\Gamma_{har}$.
	But, since $\mathfrak{P}$ is countably generated by Lemma \ref{har def}, we conclude that $\mathfrak{P}\in\mathcal{H}ar$
\end{proof}

\begin{corollary}
	The class $\mathcal{H}ar$ is P-definable.
\end{corollary}

\section{Some Variants of the Goldblatt-Thomason  Theorem}\label{sec variants}

In this section, we study an alternative version of the Goldblatt-Thomason theorem for the class of finite Markov processes with rational Markov kernels, denoted by $\mathrm{FMP}_{\mathbb{Q}_0}$.
This important subclass of Markov processes can be also seen as a natural generalization of finite Kripke structures for which there exists a version of the Goldblatt-Thomason theorem \cite{van1988notes}. So it would be natural to investigate a variant of this theorem for $\mathrm{FMP}_{\mathbb{Q}_0}$. 
So throughout this section, by a finite process we mean a Markov process $\mathfrak{P}=\langle\Omega, \mathcal{P}(\Omega), T\rangle$, where $\Omega$ is finite and 
$ T(w,w'):= T(w,\{w'\})\in\mathbb{Q}_0$.

One can associate to any finite Markov process $\mathfrak{P}=\langle\Omega, \mathcal{P}(\Omega), T\rangle$  a finite directed graph (or Kripke frame) $(\Omega,R_\mathfrak{P})$, which enables us to translate certain model theoretic concepts of probability logic into the corresponding notions of basic modal logic. 
%, where, for any state $w,w'\in\Omega$, if $T(w,w')>0$ there is an edge from $w$ to $w'$. %labeled by $r$. And if $r=0$, there is no edge from $w$ to $w'$.
%In this case, one can define point generated subspace as the one defined for Kripke frames.
%So, similar to Kripke frames, one can define Jankov-Fine like formula for finite space $\mathfrak{P}$, in a way that if this formula is satisfiable in some finite space $\mathfrak{P}'$ then there is a local bounded morphism from $\mathfrak{P}'$ to $\mathfrak{P}$. By local bounded morphism, we mean a bounded morphism which has the back and forth property up to some level $n$, and this level is related to the formula.

\begin{definition}
	Let $\mathfrak{P}=\langle\Omega, \mathcal{P}(\Omega), T\rangle$ be a finite Markov processes. Define a binary relation $R_\mathfrak{P}\subseteq\Omega\times\Omega$ as
	$$wR_\mathfrak{P}w' \; \text{ if and only if } \; T(w,w')>0.$$	
\end{definition}
Let $R_\mathfrak{P}(w)=\{w' \mid wR_\mathfrak{P}w'\}$.
Then we have $T(w,R_\mathfrak{P}(w))=1$, for each $w\in\Omega$.
Assume that $\bar{R}_\mathfrak{P}$ is the reflexive and transitive closure of $R_\mathfrak{P}$.
%Fix $w_0\in\Omega$.
Then for each $w'\in \bar{R}_\mathfrak{P}(w)$ we have $R_\mathfrak{P}(w')\subseteq \bar{R}_\mathfrak{P}(w)$. So, $T(w',\bar{R}_\mathfrak{P}(w))=1$, for each $w'\in \bar{R}_\mathfrak{P}(w)$.
Hence one can consider the sub-processes 
$$\mathfrak{P}_w=\langle \bar{R}_\mathfrak{P}(w), \mathcal{P}(\bar{R}_\mathfrak{P}(w)) , T\upharpoonright_{\bar{R}_\mathfrak{P}(w)}\rangle.$$
$\mathfrak{P}_w$ serves as a point-generated sub-frame of $(\Omega,R_\mathfrak{P})$.

A Markov process $\mathfrak{P}$ is called a point-generated process if there is $w\in \Omega$ such that $\mathfrak{P}=\mathfrak{P}_w$.

Note that if $\mathfrak{P}_w$ is a point generated finite process, then 
$\mathfrak{P}_w\models\varphi$ if and only if $\mathfrak{P}_w,w\models L_1^k\varphi$ for some $k\geq 0$.

Recall that for directed graph $(\Omega, R_\mathfrak{P})$ and $w,v\in\Omega$ the \textit{distance} between $w,v$, denoted by $d(w,v)$, is the length of the shortest path from $w$ to $v$ if there exists a path between them and $\infty$ if there is no such path.
Then, for $n\geq0$ by ${R}_\mathfrak{P}(w)\upharpoonright_n$ we mean the set of all $v\in\Omega$ with $d(w,v)\leq n$.

Below we define the notion of a local $n$-zigzag morphism. Before defining this notion notice that for two finite processes $\mathfrak{P}=\langle\Omega,\mathcal{P}(\Omega),T\rangle$ and $\mathfrak{P}'= \langle\Omega',\mathcal{P}(\Omega'),T'\rangle$ the function $f:\Omega\to\Omega'$ is a zigzag morphism if it is surjective and for each $z\in\Omega$ and $z'\in\Omega'$ we have
\begin{equation}\label{n zig equ}
	T(z,f^{-1}(z')) = T'(f(z),z').	
\end{equation}

\begin{remark}\label{point generate}
	Also note that any finite Markov process $\mathfrak{P}$ is a zigzag image of the disjoint union of its point-generated sub-processes, i.e., there is a zigzag morphism
	$f:\biguplus_{w\in\Omega}\mathfrak{P}_w \to \mathfrak{P}$, 
	%with $f=\biguplus_{w \in \Omega}f_w$ where $f_w:\bar{R}_\mathfrak{P}(w)\to \Omega$ is the inclusion function, for each $w\in\Omega$. %is defined as $f_w(v)=v$ for each $v\in\bar{R}_\mathfrak{P}(w)$.	
    where for each $v\in \biguplus_{w\in\Omega}\mathfrak{P}_w$,  $f(v)$ is equal to its corresponding state in $\Omega$.
\end{remark}

\begin{definition}[Local $n$-Zigzag]
	Let $\mathfrak{P}_w$ and $\mathfrak{P}'_{w'}$ be two finite point generated processes. A \textit{local $n$-zigzag function}, for $n\geq 0$ is a function $f:\Omega\to\Omega'$ such that ${R'}(w')\upharpoonright_n \subseteq Img(f)$ and
	(\ref{n zig equ}) holds for all $z\in {R}(w)\upharpoonright_n$ and $z'\in {R'}(w')\upharpoonright_n$.
	
	Likewise a local $n$-zigzag morphism from $\mathfrak{M}=\langle\mathfrak{P}_w,v\rangle$ to $\mathfrak{M}'=\langle\mathfrak{P}'_{w'},v'\rangle$ is a local $n$-zigzag morphism $f$ from $\mathfrak{P}_w$ to $\mathfrak{P}'_{w'}$ such that $v(p)= f^{-1}(v'(p))$, for each $p\in\mathbb{P}$.
\end{definition}

We recall that $deg_p(\varphi)\in\mathbb{N}$ is defined inductively as follows:
\begin{itemize}
	\item $deg_p(q) =0$, for $q\in\mathbb{P}$,
	\item $deg_p(\neg\varphi) = deg_p(\varphi)$,
	\item $deg_p(\varphi\wedge\psi) = \max{(deg_p(\varphi), deg_p(\psi))}$,
	\item $deg_p(L_r\varphi) = deg_p(\varphi)+1$.
\end{itemize}

\begin{lemma}\label{n-zigzag lem}
	Assume that $\mathfrak{M}_w=\langle\Omega,\mathcal{P}(\Omega),T,v\rangle$ and $\mathfrak{M}'_{w'}=\langle \Omega',\mathcal{P}(\Omega'), T',v'\rangle$ are two point-generated models and $f:\Omega\to \Omega'$ is a local $n$-zigzag morphism. Then 
	\[
	\mathfrak{M},w\models\psi \; \text{ if and only if } \; \mathfrak{M}',f(w)\models\psi,	
	\]
	for each $\psi$ with $deg_p(\psi)\leq n$.
\end{lemma}
\begin{proof}%[Proof of Lemma \ref{n-zigzag lem}]
	We prove it by induction on $n$ and complexity of $\psi$.
	We only prove the induction step for $L_r\theta$, assuming $deg_p(\theta)\leq n$.
	Suppose that $f$ is a local $(n+1)$-zigzag morphism. Then by induction hypothesis, we have
	$$\llbracket\theta\rrbracket_{\mathfrak{M}}\cap  R(w) = f^{-1}(\llbracket\theta\rrbracket_{\mathfrak{M}'}\cap  R'(f(w))),$$
	since for each $z\in R(w)$, the function $f\upharpoonright_{\bar{R}(z)}: \bar{R}(z) \to \bar{R}'(f(z)) $ is a local $n$-zigzag morphism between $\mathfrak{P}_z$ and $\mathfrak{P}'_{f(z)}$.
	Hence,  
	$$T(w,\llbracket\theta\rrbracket_{\mathfrak{M}}) = T(w,\llbracket\theta\rrbracket_{\mathfrak{M}}\cap  R(w))  = T(f(w), \llbracket\theta\rrbracket_{\mathfrak{M}'}\cap  R'(f(w)) ) = T(f(w),\llbracket\theta\rrbracket_{\mathfrak{M}'}).$$
	But this implies 
	$$\mathfrak{M},w\models L_r\theta \; \text{ if and only if } \mathfrak{M}',f(w)\models L_r\theta.$$	
\end{proof}

Below we prove a version of the Goldblatt-Thomason theorem for a class of finite Markov processes with rational probabilities.
Before stating the theorem we isolate an important property which is necessary for a subclass $\mathcal{C}$ of $\mathrm{FMP}_{\mathbb{Q}_0}$ to be definable.
\begin{definition}
	Let $\mathcal{C}$ be a subclass of  $\mathrm{FMP}_{\mathbb{Q}_0}$. We say that $\mathcal{C}$ is closed under local zigzag morphism if for each point-generated process $\mathfrak{P}_w$ we have $\mathfrak{P}_w\in\mathcal{C}$ whenever for each $n>0$ there is a point-generated process $\mathfrak{P}'_n\in\mathcal{C}$ where there is a local $n$-zigzag from $\mathfrak{P}'_n$ to $\mathfrak{P}$.
\end{definition}
%%
%%

%Then the theory of $\mathcal{C}$ is valid on $\mathfrak{P}$.

\begin{proposition}\label{n-zigzag prop}
	Let $\Gamma$ be a set of formulas and assume that $\mathcal{C}$ is the class of all finite Markov processes $\mathfrak{P}$ in which  $\Gamma$ is valid. Then $\mathcal{C}$ is closed under local zigzag morphisms.
\end{proposition}
\begin{proof}%[Proof of Proposition \ref{n-zigzag prop}]
	Assume that $\mathfrak{P}_w= \langle\Omega_w,\mathcal{P}(\Omega_w),T\rangle$ is a point generated Markov process in $\mathrm{FMP}_{\mathbb{Q}_0}$ such that it is an image of local zigzag morphisms.
	We have to show that $\mathfrak{P}_w\models\Gamma$. 
	Assume not. Then for some $\varphi\in\Gamma$ we have $\mathfrak{P}_w\not\models\varphi$.
	Hence there exist $\mathfrak{M}=\langle \mathfrak{P}_w,v\rangle$ and $k\geq 0$ such that $\mathfrak{M},w\not\models L_1^k\varphi$.
	Now take $n\geq \max{(|\Omega_w|, deg_p(L_1^k\varphi))}$, $\mathfrak{P}_{w'}\in\mathcal{C}$  and $f:\mathfrak{P}_{w'}\to \mathfrak{P}_w$ a local $n$-zigzag morphism.
	By the choice of $n$, we conclude that $f$ is a surjection.
	Moreover, we can define $v'$ on $\mathfrak{P}_{w'}$ such that $f$ is a local $n$-zigzag morphism from $\mathfrak{M}'=\langle\mathfrak{P}_{w'},v'\rangle$ to $\mathfrak{M}=\langle\mathfrak{P}_w,v\rangle$.
	Hence in the light of Lemma \ref{n-zigzag lem}, we have $\mathfrak{M}',w'\not\models L_1^k\varphi$. Therefore, $\mathfrak{P}_{w'}\not\models\varphi$, a contradiction.
\end{proof}

\begin{theorem}
	Assume that $\mathcal{C}$ is a subclass of $\mathrm{FMP}_{\mathbb{Q}_0}$.
	Then $\mathcal{C}$ is definable if and only if it is closed under generated sub-processes, finite disjoint unions, and local zigzag morphisms.
\end{theorem}
\begin{proof}
	The nontrivial direction is the right-to-left one, so assume that $\mathfrak{P}$ is a finite process with $\mathfrak{P}\models Th(\mathcal{C})$. We have to show that $\mathfrak{P}\in\mathcal{C}$.
	
	By Remark \ref{point generate} and since $\mathcal{C}$ is closed under finite disjoint unions, without loss of generality, we can assume that $\mathfrak{P}$ is point-generated.
	Suppose $\Omega=\{w_0,w_1,\dots,w_m\}$ where $\mathfrak{P}=\mathfrak{P}_{w_0}$. Let $n>0$ be arbitrary. %the largest distance from $w_0$ in $\mathfrak{P}$.
	Consider the language $\mathcal{L}_\mathfrak{P}=\{p_0,\dots, p_m\}$, and the valuation $v$ on $\mathfrak{P}$ as  $v(p_i) =\{w_i\}$, for each $1\leq i\leq n$.
	Define the following formulas:
	\begin{enumerate}[i.]
		\item $p_{_0}$.
		\item $L_r p_{_i} \wedge M_r p_{_i}$, if $T(w_0,w_i)=r$.
		\item $L_1^k(p_0\vee\dots\vee p_m)$, for $0\leq k\leq n$.
		\item $L_1^k (p_i\to \neg p_j)$, for $0\leq k\leq n$ and $i\neq j$.
		\item $L_1^k(p_i\to (L_r p_{_j} \wedge M_r p_{_j}))$, for $0\leq k\leq n$ if $T(w_i,w_j)=r$.
	\end{enumerate}
	While (i) and (ii) are satisfied in $w_0$, the other ones are valid on $\mathfrak{M}$, and in particular, are satisfied in $w_0$.
	Let $\varphi_{n,w_0}$ be the conjunction of the above formulas. Clearly,
	$\varphi_{n,w_0}\in Th(\mathcal{C})$ and is satisfiable in $\mathcal{C}$, as otherwise $\neg\varphi_{n,w_0}\in Th(\mathcal{C})$, which is a contradiction.
	So there exists $\mathfrak{N}_n=\langle\mathfrak{P}'_n,v'\rangle$ with $\mathfrak{P}'_n=\langle\Omega',\mathcal{P}(\Omega'),T'\rangle\in\mathcal{C}$ such that $\mathfrak{N}_n,z_0\models\varphi_{n,w_0}$ for some   $z_0\in\Omega'$.
	Again, without loss of generality, we can assume that $\mathfrak{P}'$ is point-generated with $\mathfrak{P}'=\mathfrak{P}'_{z_0}$.
	
	Now define $f_n:\Omega'\to \Omega$ as $f_n(z) =w_i$ if $z\in R(v_0)\upharpoonright_n$ and  $\mathfrak{N}_n,z\models p_i$. Otherwise let $f_n(z)=w_0$. We show that $f_n$ is a local $n$-zigzag morphism.
	\begin{itemize}
		\item $f_n$ is well-defined:
		assume that $z\in R(z_0)\upharpoonright_n$. Hence $d(v_0,z)=k\leq n$, and $\mathfrak{N}_n,v_0\models L_1^k (p_0\vee\dots \vee p_m)$.
		Therefore, this implies that $\mathfrak{N}_n,z\models (p_0\vee\dots \vee p_m)$.
		Thus, there is $0\leq i\leq m$ such that $\mathfrak{N}_n,z\models p_i$.
		Furthermore, for $i\neq j$ we have $\mathfrak{N}_n,z_0\models L_1^k(p_i\to\neg p_j)$ which implies that $\mathfrak{N}_n,z\models p_i\to \neg p_j$. So we conclude that there is a unique $i$ such that $f_n(z)=w_i$.
		
		\item $Img(f_n)\subseteq R(w_0)\upharpoonright_n$: %$\mathfrak{P}=\bar{R}(w_0)=R(w_0)\upharpoonright_n$.
		Assume that $w\in\Omega$ and $d(w_0,w)=k\leq n$.  
		So there are $w_{j_1},\dots, w_{j_k}\in\Omega$ with $w_0Rw_{j_1}\dots Rw_{j_k}$ and $w_{j_k}=w$.
		Suppose that $T(w_{j_i},w_{j_{i+1}})=r_i>0$, for each $0\leq i\leq k-1$.
		Since $\mathfrak{M},w_0\models L_{r_1}p_{j_1} \wedge L_1(p_{j_1}\to L_{r_2}p_{j_2})$, it follows that  $\mathfrak{N}_n,v_0\models L_{r_1}p_{j_1} \wedge L_1(p_{j_1}\to L_{r_2}p_{j_2})$.
		Hence, 		
		$T'(z_0,\llbracket p_{j_1}\rrbracket_{\mathfrak{N}_n})=r_1>0$
		and $T'(z_0,\llbracket p_{j_1}\to L_{r_2}p_{j_2}\rrbracket_{\mathfrak{N}_n})=1$. So there is $z_1\in\Omega'$ such that $T'(z_0,z_1)>0$ and $\mathfrak{N}_n,z_1\models p_{i_1}\wedge L_{r_2}p_{j_2}$.
		Again, since 
		$\mathfrak{M},w_0\models L_1^2(p_{j_2}\to L_{r_3}p_{j_3})$, it follows that $T'(z_1,\llbracket p_{j_2}\to L{r_3}p_{j_3}\rrbracket_{\mathfrak{N}_n})=1$.
		Thus there is $z_2\in\Omega'$ such that $\mathfrak{N},z_2\models p_{i_2}$ and $T'(z_1,z_2)>0$. Continuing this method, one can find a sequence $z_0R'z_1\dots R'z_{i}$ such that $\mathfrak{N}_n,z_i\models p_{j_i}$ for each $0<i\leq k$.
		So, $f_n(z_k)=w$.
		\item $f_n$ is a local $n$-zigzag morphism: 
		note that the measurability of $f_n$ is automatic. So, we only need to show that
		$$T'(z,f_n^{-1}(w)) = T(f_n(z),w),$$
		for each $w\in R(w)\upharpoonright_n$ and $z\in R'(z_0)\upharpoonright_n$. 
		Assume that for some $0\leq i,j\leq m$ we have $\mathfrak{N},z\models p_i$ and $\mathfrak{M},w\models p_j$.
		Thus $f(z) = w_i$. If $T(w,w_i)=r$, then for each $0\leq k\leq n$, 
		$$\mathfrak{M},w_0\models L_1^k(p_j\to (L_r p_i\wedge M_r p_i)).$$ 
		So 
		$$\mathfrak{N}_n,z_0\models L_1^k(p_j\to (L_r p_i\wedge M_r p_i)).$$
		Hence, 
		$$T'(z,\llbracket p_i\rrbracket_{\mathfrak{N}_n})=T'(z,f_n^{-1}(w))=r.$$
	\end{itemize}
	Since $\mathfrak{P}'_n\in\mathcal{C}$ and $f_n:\mathfrak{P}'_n\to\mathfrak{P}$ is a local $n$-zigzag morphism, it follows that $\mathfrak{P}\in\mathcal{C}$, as required.
\end{proof}
%%
%%
%%%%%%%%%%%%%%%%%%%%%%%%%%%%%%%%%%%%%%%%
%%%%%%%%%%%%%%%%%%%%%%%%%%%%%%%%%%%%%%%%
\section{Conclusions and further studies}
In this paper, we established a Goldblatt-Thomason-style theorem for probability logic, characterizing the classes of Markov processes definable in probability logic via suitable model-theoretic closure properties. We also obtained a corresponding result for finite Markov processes. These results connect probabilistic modal logic with the classical model-theoretic theory of definability.

A natural direction for future research is to extend these results to infinitary probability logics, which allow countable conjunctions and disjunctions and whose expressive power is closely related to the $\sigma$-additivity of probability measures. As shown in \cite{ChopoghlooPourmahdian2024}, important stochastic properties of Markov processes, such as stationarity, invariance, irreducibility, and recurrence, are expressible in this setting. This indicates that infinitary probability logics provide a more suitable framework for capturing genuinely measure-theoretic probabilistic phenomena. Accordingly, an important problem for future work is to establish a Goldblatt-Thomason-style theorem for infinitary probability logic and to characterize the associated model-theoretic closure conditions for definability.

\nocite{*}
%\bibliographystyle{eptcs}
%\bibliography{generic}

\begin{thebibliography}{10}
\providecommand{\bibitemdeclare}[2]{}
\providecommand{\surnamestart}{}
\providecommand{\surnameend}{}
\providecommand{\urlprefix}{Available at }
\providecommand{\url}[1]{\texttt{#1}}
\providecommand{\href}[2]{\texttt{#2}}
\providecommand{\urlalt}[2]{\href{#1}{#2}}
\providecommand{\doi}[1]{doi:\urlalt{https://doi.org/#1}{#1}}
\providecommand{\eprint}[1]{arXiv:\urlalt{https://arxiv.org/abs/#1}{#1}}
\providecommand{\bibinfo}[2]{#2}

\bibitemdeclare{article}{Aum1999}
\bibitem{Aum1999}
\bibinfo{author}{Robert~J. \surnamestart Aumann\surnameend}
  (\bibinfo{year}{1999}): \emph{\bibinfo{title}{Interactive Epistemology II:
  Probability}}.
\newblock {\slshape \bibinfo{journal}{International Journal of Game Theory}}
  \bibinfo{volume}{28}(\bibinfo{number}{3}), pp. \bibinfo{pages}{301--314},
  \doi{10.1007/s001820050112}.

\bibitemdeclare{article}{van1988notes}
\bibitem{van1988notes}
\bibinfo{author}{Johan \surnamestart van Benthem\surnameend}
  (\bibinfo{year}{1988}): \emph{\bibinfo{title}{Notes on Modal Definability}}.
\newblock {\slshape \bibinfo{journal}{Notre Dame Journal of Formal Logic}}
  \bibinfo{volume}{30}(\bibinfo{number}{1}), pp. \bibinfo{pages}{20--35},
  \doi{10.3233/FI-1993-182-416}.

\bibitemdeclare{article}{benthem:revisited93}
\bibitem{benthem:revisited93}
\bibinfo{author}{Johan \surnamestart van Benthem\surnameend}
  (\bibinfo{year}{1993}): \emph{\bibinfo{title}{Modal Frame Classes
  Revisited}}.
\newblock {\slshape \bibinfo{journal}{Fundamenta Informaticae}}
  \bibinfo{volume}{18}, pp. \bibinfo{pages}{307--317},
  \doi{10.3233/FI-1993-182-416}.

\bibitemdeclare{inproceedings}{cardelli2011continuous}
\bibitem{cardelli2011continuous}
\bibinfo{author}{Luca \surnamestart Cardelli\surnameend},
  \bibinfo{author}{Kim~G. \surnamestart Larsen\surnameend} \&
  \bibinfo{author}{Radu \surnamestart Mardare\surnameend}
  (\bibinfo{year}{2011}): \emph{\bibinfo{title}{Continuous Markovian Logic:
  From Complete Axiomatization to the Metric Space of Formulas}}.
\newblock In: {\slshape \bibinfo{booktitle}{Computer Science Logic 2011 (CSL
  2011)}}, {\slshape \bibinfo{series}{Leibniz International Proceedings in
  Informatics (LIPIcs)}}~\bibinfo{volume}{12}, \bibinfo{publisher}{Schloss
  Dagstuhl--Leibniz-Zentrum f{\"u}r Informatik}, pp. \bibinfo{pages}{144--158},
  \doi{10.4230/LIPIcs.CSL.2011.144}.

\bibitemdeclare{article}{ChopoghlooPourmahdian2024}
\bibitem{ChopoghlooPourmahdian2024}
\bibinfo{author}{Somayeh \surnamestart Chopoghloo\surnameend} \&
  \bibinfo{author}{Massoud \surnamestart Pourmahdian\surnameend}
  (\bibinfo{year}{2024}): \emph{\bibinfo{title}{Dynamic Probability Logics:
  Axiomatization \& Definability}}.
\newblock {\slshape \bibinfo{journal}{arXiv preprint}}
  \bibinfo{volume}{arXiv:2401.07235}, \doi{10.48550/arXiv.2401.07235}.

\bibitemdeclare{article}{deshar:bisim02}
\bibitem{deshar:bisim02}
\bibinfo{author}{Jos{\'e}e \surnamestart Desharnais\surnameend},
  \bibinfo{author}{Abbas \surnamestart Edalat\surnameend} \&
  \bibinfo{author}{Prakash \surnamestart Panangaden\surnameend}
  (\bibinfo{year}{2002}): \emph{\bibinfo{title}{Bisimulation for Labelled
  Markov Processes}}.
\newblock {\slshape \bibinfo{journal}{Information and Computation}}
  \bibinfo{volume}{179}(\bibinfo{number}{2}), pp. \bibinfo{pages}{163--193},
  \doi{10.1006/inco.2001.2962}.

\bibitemdeclare{article}{fagin:logic90}
\bibitem{fagin:logic90}
\bibinfo{author}{Ronald \surnamestart Fagin\surnameend},
  \bibinfo{author}{Joseph~Y. \surnamestart Halpern\surnameend} \&
  \bibinfo{author}{Nimrod \surnamestart Megiddo\surnameend}
  (\bibinfo{year}{1990}): \emph{\bibinfo{title}{A Logic for Reasoning about
  Probabilities}}.
\newblock {\slshape \bibinfo{journal}{Information and Computation}}
  \bibinfo{volume}{87}(\bibinfo{number}{1--2}), pp. \bibinfo{pages}{78--128},
  \doi{10.1016/0890-5401(90)90060-U}.

\bibitemdeclare{inproceedings}{fornasiere2024frame}
\bibitem{fornasiere2024frame}
\bibinfo{author}{Damiano \surnamestart Fornasiere\surnameend},
  \bibinfo{author}{Johannes \surnamestart Marti\surnameend} \&
  \bibinfo{author}{Giovanni \surnamestart Varricchione\surnameend}
  (\bibinfo{year}{2024}): \emph{\bibinfo{title}{Frame Definability in
  Conditional Logic}}.
\newblock In: {\slshape \bibinfo{booktitle}{Advances in Modal Logic}},
  \bibinfo{volume}{15}, \bibinfo{publisher}{College Publications}, pp.
  \bibinfo{pages}{335--356}.
\newblock
  \urlprefix\url{http://www.aiml.net/volumes/volume15/Fornasiere-Marti-Varricchione.pdf}.

\bibitemdeclare{article}{gt:axiom75}
\bibitem{gt:axiom75}
\bibinfo{author}{Robert~I. \surnamestart Goldblatt\surnameend} \&
  \bibinfo{author}{Stephen~K. \surnamestart Thomason\surnameend}
  (\bibinfo{year}{1975}): \emph{\bibinfo{title}{Axiomatic Classes in
  Propositional Modal Logic}}.
\newblock {\slshape \bibinfo{journal}{Algebra and Logic}} \bibinfo{volume}{13},
  pp. \bibinfo{pages}{163--173}, \doi{10.1007/BFb0062855}.

\bibitemdeclare{inproceedings}{de2020goldblatt}
\bibitem{de2020goldblatt}
\bibinfo{author}{Jim \surnamestart de~Groot\surnameend} (\bibinfo{year}{2022}):
  \emph{\bibinfo{title}{Goldblatt-Thomason Theorems for Modal Intuitionistic
  Logics}}.
\newblock In: {\slshape \bibinfo{booktitle}{Advances in Modal Logic}},
  \bibinfo{volume}{14}, \bibinfo{publisher}{College Publications}, pp.
  \bibinfo{pages}{467--489}.
\newblock \urlprefix\url{http://www.aiml.net/volumes/volume14/28-deGroot.pdf}.

\bibitemdeclare{incollection}{harsan:games68}
\bibitem{harsan:games68}
\bibinfo{author}{John~C. \surnamestart Harsanyi\surnameend}
  (\bibinfo{year}{1982}): \emph{\bibinfo{title}{Games with Incomplete
  Information Played by ``Bayesian'' Players, Part II: Bayesian Equilibrium
  Points}}.
\newblock In: {\slshape \bibinfo{booktitle}{Papers in Game Theory}},
  \bibinfo{publisher}{Springer Netherlands}, pp. \bibinfo{pages}{139--153},
  \doi{10.1007/978-94-017-2527-9_7}.

\bibitemdeclare{article}{heifmon:prob01}
\bibitem{heifmon:prob01}
\bibinfo{author}{Aviad \surnamestart Heifetz\surnameend} \&
  \bibinfo{author}{Philippe \surnamestart Mongin\surnameend}
  (\bibinfo{year}{2001}): \emph{\bibinfo{title}{Probability Logic for Type
  Spaces}}.
\newblock {\slshape \bibinfo{journal}{Games and Economic Behavior}}
  \bibinfo{volume}{35}(\bibinfo{number}{1}), pp. \bibinfo{pages}{31--53},
  \doi{10.1006/game.1999.0788}.

\bibitemdeclare{book}{Johnstone1982}
\bibitem{Johnstone1982}
\bibinfo{author}{Peter~T. \surnamestart Johnstone\surnameend}
  (\bibinfo{year}{1982}): \emph{\bibinfo{title}{Stone Spaces}}.
\newblock {\slshape \bibinfo{series}{Cambridge Studies in Advanced
  Mathematics}}~\bibinfo{volume}{3}, \bibinfo{publisher}{Cambridge University
  Press}, \bibinfo{address}{Cambridge}.
\newblock \urlprefix\url{https://books.google.com/books?id=CiWwoLNbpykC}.

\bibitemdeclare{incollection}{keisler}
\bibitem{keisler}
\bibinfo{author}{H.~Jerome \surnamestart Keisler\surnameend}
  (\bibinfo{year}{1985}): \emph{\bibinfo{title}{Probability Quantifiers}}.
\newblock In \bibinfo{editor}{Jon \surnamestart Barwise\surnameend} \&
  \bibinfo{editor}{Solomon \surnamestart Feferman\surnameend}, editors:
  {\slshape \bibinfo{booktitle}{Model-Theoretic Logics}}, {\slshape
  \bibinfo{series}{Perspectives in Mathematical Logic}}~\bibinfo{volume}{8},
  \bibinfo{publisher}{Springer}, pp. \bibinfo{pages}{509--556},
  \doi{10.1007/978-1-4613-8928-6_9}.

\bibitemdeclare{inproceedings}{kozen2013stone}
\bibitem{kozen2013stone}
\bibinfo{author}{Dexter \surnamestart Kozen\surnameend},
  \bibinfo{author}{Kim~G. \surnamestart Larsen\surnameend},
  \bibinfo{author}{Radu \surnamestart Mardare\surnameend} \&
  \bibinfo{author}{Prakash \surnamestart Panangaden\surnameend}
  (\bibinfo{year}{2013}): \emph{\bibinfo{title}{Stone Duality for Markov
  Processes}}.
\newblock In: {\slshape \bibinfo{booktitle}{Proceedings of the 28th Annual
  ACM/IEEE Symposium on Logic in Computer Science (LICS 2013)}},
  \bibinfo{publisher}{IEEE}, pp. \bibinfo{pages}{321--330},
  \doi{10.1109/LICS.2013.38}.

\bibitemdeclare{inproceedings}{kozen2013strong}
\bibitem{kozen2013strong}
\bibinfo{author}{Dexter \surnamestart Kozen\surnameend}, \bibinfo{author}{Radu
  \surnamestart Mardare\surnameend} \& \bibinfo{author}{Prakash \surnamestart
  Panangaden\surnameend} (\bibinfo{year}{2013}): \emph{\bibinfo{title}{Strong
  Completeness for Markovian Logics}}.
\newblock In: {\slshape \bibinfo{booktitle}{Mathematical Foundations of
  Computer Science 2013}}, {\slshape \bibinfo{series}{Lecture Notes in Computer
  Science}} \bibinfo{volume}{8087}, \bibinfo{publisher}{Springer}, pp.
  \bibinfo{pages}{655--666}, \doi{10.1007/978-3-642-40313-2_58}.

\bibitemdeclare{inproceedings}{kurz:golcoal07}
\bibitem{kurz:golcoal07}
\bibinfo{author}{Alexander \surnamestart Kurz\surnameend} \&
  \bibinfo{author}{Ji\v{r}\'{\i} \surnamestart Rosick{\'y}\surnameend}
  (\bibinfo{year}{2007}): \emph{\bibinfo{title}{The Goldblatt-Thomason Theorem
  for Coalgebras}}.
\newblock In \bibinfo{editor}{Till \surnamestart Mossakowski\surnameend},
  \bibinfo{editor}{Ugo \surnamestart Montanari\surnameend} \&
  \bibinfo{editor}{Magne \surnamestart Haveraaen\surnameend}, editors:
  {\slshape \bibinfo{booktitle}{Algebra and Coalgebra in Computer Science
  (CALCO 2007)}}, {\slshape \bibinfo{series}{Lecture Notes in Computer
  Science}} \bibinfo{volume}{4624}, \bibinfo{publisher}{Springer}, pp.
  \bibinfo{pages}{342--355}, \doi{10.1007/978-3-540-73859-6_23}.

\bibitemdeclare{article}{kuter:model13}
\bibitem{kuter:model13}
\bibinfo{author}{Rutger \surnamestart Kuyper\surnameend} \&
  \bibinfo{author}{Sebastiaan~A. \surnamestart Terwijn\surnameend}
  (\bibinfo{year}{2013}): \emph{\bibinfo{title}{Model Theory of Measure Spaces
  and Probability Logic}}.
\newblock {\slshape \bibinfo{journal}{Review of Symbolic Logic}}
  \bibinfo{volume}{6}, pp. \bibinfo{pages}{367--393},
  \doi{10.1017/S1755020313000063}.

\bibitemdeclare{article}{ma2025goldblatt}
\bibitem{ma2025goldblatt}
\bibinfo{author}{Minghui \surnamestart Ma\surnameend} (\bibinfo{year}{2025}):
  \emph{\bibinfo{title}{Goldblatt-Thomason Theorems for Inflationary
  Intuitionistic Logic}}.
\newblock {\slshape \bibinfo{journal}{Studia Logica}}, pp.
  \bibinfo{pages}{1--30}, \doi{10.1007/s11225-025-10214-9}.

\bibitemdeclare{book}{panan:lmp09}
\bibitem{panan:lmp09}
\bibinfo{author}{Prakash \surnamestart Panangaden\surnameend}
  (\bibinfo{year}{2009}): \emph{\bibinfo{title}{Labelled Markov Processes}}.
\newblock \bibinfo{publisher}{Imperial College Press}.
\newblock \urlprefix\url{https://books.google.com/books?id=q-ZpDQAAQBAJ}.

\bibitemdeclare{article}{RasiowaSikorski1950}
\bibitem{RasiowaSikorski1950}
\bibinfo{author}{Helena \surnamestart Rasiowa\surnameend} \&
  \bibinfo{author}{Roman \surnamestart Sikorski\surnameend}
  (\bibinfo{year}{1950}): \emph{\bibinfo{title}{A proof of the completeness
  theorem of G{\"o}del}}.
\newblock {\slshape \bibinfo{journal}{Fundamenta Mathematicae}}
  \bibinfo{volume}{37}, pp. \bibinfo{pages}{193--200},
  \doi{10.4064/fm-37-1-193-200}.

\bibitemdeclare{phdthesis}{rodenburg2016intuitionistic}
\bibitem{rodenburg2016intuitionistic}
\bibinfo{author}{Piet~H. \surnamestart Rodenburg\surnameend}
  (\bibinfo{year}{1986}): \emph{\bibinfo{title}{Intuitionistic Correspondence
  Theory}}.
\newblock \bibinfo{type}{Phd thesis}, \bibinfo{school}{University of
  Amsterdam}.
\newblock \urlprefix\url{https://eprints.illc.uva.nl/id/eprint/1848/}.

\bibitemdeclare{inproceedings}{sano2020goldblatt}
\bibitem{sano2020goldblatt}
\bibinfo{author}{Katsuhiko \surnamestart Sano\surnameend},
  \bibinfo{author}{David \surnamestart Fern{\'a}ndez-Duque\surnameend},
  \bibinfo{author}{Alessandra \surnamestart Palmigiano\surnameend} \&
  \bibinfo{author}{Sophie \surnamestart Pinchinat\surnameend}
  (\bibinfo{year}{2020}): \emph{\bibinfo{title}{Goldblatt-Thomason-Style
  Characterization for Intuitionistic Inquisitive Logic}}.
\newblock In: {\slshape \bibinfo{booktitle}{Advances in Modal Logic}},
  \bibinfo{volume}{13}, \bibinfo{publisher}{College Publications}, pp.
  \bibinfo{pages}{541--560}.
\newblock
  \urlprefix\url{http://www.aiml.net/volumes/volume13/Sano-FernandezDuque-Palmigiano-Pinchinat.pdf}.

\bibitemdeclare{inproceedings}{sano:gtgraded10}
\bibitem{sano:gtgraded10}
\bibinfo{author}{Katsuhiko \surnamestart Sano\surnameend} \&
  \bibinfo{author}{Minghui \surnamestart Ma\surnameend} (\bibinfo{year}{2010}):
  \emph{\bibinfo{title}{Goldblatt-Thomason-Style Theorems for Graded Modal
  Language}}.
\newblock In \bibinfo{editor}{Lev \surnamestart Beklemishev\surnameend},
  \bibinfo{editor}{Valentin \surnamestart Goranko\surnameend} \&
  \bibinfo{editor}{Valentin \surnamestart Shehtman\surnameend}, editors:
  {\slshape \bibinfo{booktitle}{Advances in Modal Logic}}, \bibinfo{volume}{8},
  \bibinfo{publisher}{College Publications}, pp. \bibinfo{pages}{330--349}.
\newblock \urlprefix\url{http://www.aiml.net/volumes/volume8/Sano-Ma.pdf}.

\bibitemdeclare{article}{zhou:hars14}
\bibitem{zhou:hars14}
\bibinfo{author}{Chunlai \surnamestart Zhou\surnameend} (\bibinfo{year}{2014}):
  \emph{\bibinfo{title}{Probability Logic for Harsanyi Type Spaces}}.
\newblock {\slshape \bibinfo{journal}{Logical Methods in Computer Science}}
  \bibinfo{volume}{10}(\bibinfo{number}{2}), \doi{10.2168/LMCS-10(2:13)2014}.

\end{thebibliography}

\end{document}